\documentclass[sn-mathphys-num]{sn-jnl}

\usepackage{graphicx}%
\usepackage{multirow}%
\usepackage{amsmath,amssymb,amsfonts}%
\usepackage{amsthm}%
\usepackage{mathrsfs}%
\usepackage{appendix}%
\usepackage{xcolor}%
\usepackage{textcomp}%
\usepackage{manyfoot}%
\usepackage{booktabs}%
\usepackage[noend]{algorithmic}%
\usepackage{algorithm}%
\usepackage{listings}%
\usepackage{afterpage}
\usepackage{enumerate}

\usepackage{array}
\newcolumntype{N}{>{\scriptsize\centering\arraybackslash}c} 
\usepackage{makecell}
\setcellgapes{2pt}

\newtheorem{prop}{Proposition}
\newtheorem{cor}{Corollary}

\newtheorem{rem}{Remark}

\newtheorem{conjecture}{Conjecture}

\usepackage{makecell}

\usepackage{enumerate}
\usepackage{accents}

\usepackage{mathdots}
\usepackage{booktabs}
\usepackage{mathtools}
\usepackage{multirow}
\usepackage{makecell}

\usepackage{pgfplots}
\pgfplotsset{compat=newest} 

\usepackage{booktabs}
\usepackage{lscape}  

\usepackage{tikz}
\usepackage{tikz-3dplot}

\usepackage{subcaption}

\newtheorem{fact}{Fact}

\usepackage{amssymb}

\usepackage{geometry} 
\usepackage{cleveref}
\crefname{enumi}{}{\unskip}

\usepackage{verbatim}

\usepackage{appendix}
\usepackage{titlesec}

\begin{document}

\title{Solving the Heilbronn Triangle Problem using Global Optimization Methods}

\author[1]{\fnm{Amirali} \sur{Modir} 
}\email{amirali.mgh@sharif.edu}
\equalcont{These authors contributed equally to this work.}

\author[2]{\fnm{Amirhossein} \sur{Monji} 
}\email{amirhossein.monji@sharif.edu}
\equalcont{These authors contributed equally to this work.}


\author[3]{\fnm{Burak} \sur{Kocuk} 
}\email{burakkocuk@sabanciuniv.edu}

\title{
Solving the Heilbronn Triangle Problem using Global Optimization Methods
}

\affil[1,2]{\orgdiv{Industrial Engineering Department}, \orgname{Sharif University of Technology}, \orgaddress{ \postcode{1458889694}, \city{Tehran}, \country{Iran}}}

\affil[3]{\orgdiv{Industrial Engineering Program}, \orgname{Sabancı University}, \orgaddress{ \postcode{34956}, \city{Istanbul}, \country{Turkey}}}


\abstract{We study the Heilbronn triangle problem, which involves placing \(n\) points in the unit square such that the minimum area of any triangle formed by these points is maximized. A straightforward maximin formulation of this problem is highly non-linear and non-convex due to the existence of bilinear terms and absolute value equations. We propose two mixed-integer quadratically constrained programming (MIQCP) and one QCP formulation, which can be readily solved by any global optimization solver. We develop several formulation enhancements in the form of bound tightening and symmetry breaking inequalities that are prevalent in the global optimization literature in addition to other enhancements that exploit the problem structure. With the help of these enhancements, our models reproduce proven optimal values for instances up to $n=8$ points with certified optimality in the order of seconds. In the case of $n=9$ points, for which no analytical proof is known, we establish a certified optimal value by a computational effort of one day. This is a significant improvement over the previous benchmark established in 31 days of computations by Chen et al. (2017).}




\maketitle

\section{Introduction}

The Heilbronn triangle problem involves placing \(n\) points ($n>2$) in the unit square $[0,1] \times [0,1]$ such that the minimum area of any triangle formed by these points is maximized. More precisely, denoting the coordinates of point $i$ as   $(x_i,y_i) \in [0,1] \times [0,1]$  for $i=1,\dots,n$, our aim is to solve the following optimization problem:
\begin{equation}\label{eq:heilbronnGeneric}
H_n^* := \max_{ (x,y)\in[0,1]^n\times[0,1]^n} \ \min_{1\le i<j<k \le n} \left\{  \frac12|    x_i(y_j-y_k) - x_j(y_i-y_k)+x_k(y_i-y_j)     | \right\}.
\end{equation}
Notice that the expression $\frac12|    x_i(y_j-y_k) - x_j(y_i-y_k)+x_k(y_i-y_j)     |$ gives the area of the triangle formed by points $i$, $j$ and $k$. We remark that problem~\eqref{eq:heilbronnGeneric} is highly non-convex due to the existence of the bilinear terms as well as the absolute value function.
For small $n$ values such as 3, 4 or 5, the optimal configurations are trivial to obtain (see Figure~\ref{fig:n3_to_n5_placements} for an illustration). However, for slightly larger values of $n$, the Heilbronn triangle problem becomes notoriously difficult to solve. In this paper, our aim is to solve this problem up to $n=9$ points using global optimization methods.
\begin{figure}[h!]
\centering
\begin{tikzpicture}[scale=4, thick]

  \begin{scope}[shift={(0,0)}]
    \node[align=center] at (0.5,1.10)
  {{$n=3$}\\[-0.2ex]{\scriptsize proven optimal value: 0.5}};

    \draw (0,0) rectangle (1,1);

    \coordinate (A) at (0,0);
    \coordinate (B) at (1,0);
    \coordinate (C) at (0,1);

    \draw[black!65] (A)--(B)--(C)--cycle;

    \fill[blue!25, opacity=.35] (A)--(B)--(C)--cycle;
    \draw[blue!60!black]        (A)--(B)--(C)--cycle;

    \foreach \P in {A,B,C}{\fill (\P) circle (0.015);}
  \end{scope}

  \begin{scope}[shift={(1.35,0)}]
        \node[align=center] at (0.5,1.10)
  {{$n=4$}\\[-0.2ex]{\scriptsize proven optimal value: 0.5}};

    \draw (0,0) rectangle (1,1);

    \coordinate (A) at (0,0);
    \coordinate (B) at (1,0);
    \coordinate (C) at (0,1);
    \coordinate (D) at (1,1);

    \draw[black!65] (A)--(B)--(C)--cycle; 
    \draw[black!65] (A)--(B)--(D)--cycle; 
    \draw[black!65] (A)--(C)--(D)--cycle; 
    \draw[black!65] (B)--(C)--(D)--cycle; 

    \fill[green!25, opacity=.35] (A)--(B)--(C)--cycle;
    \draw[green!55!black]         (A)--(B)--(C)--cycle;

    \foreach \P in {A,B,C,D}{\fill (\P) circle (0.015);}
  \end{scope}

  \begin{scope}[shift={(2.70,0)}]
        \node[align=center] at (0.5,1.10)
  {{$n=5$}\\[-0.2ex]{\scriptsize proven optimal value: 0.19245}};

    \draw (0,0) rectangle (1,1);

    \coordinate (P1) at (0.3339240927352503, 0.0);
    \coordinate (P2) at (1.0, 0.0);
    \coordinate (P3) at (0.0, 0.5778621196377186);
    \coordinate (P4) at (1.0, 0.6660759072647495);
    \coordinate (P5) at (0.4221376453891849, 1.0);

    \draw[black!65] (P1)--(P2)--(P3)--cycle;
    \draw[black!65] (P1)--(P2)--(P4)--cycle;
    \draw[black!65] (P1)--(P2)--(P5)--cycle;
    \draw[black!65] (P1)--(P3)--(P4)--cycle;
    \draw[black!65] (P1)--(P3)--(P5)--cycle; 
    \draw[black!65] (P1)--(P4)--(P5)--cycle;
    \draw[black!65] (P2)--(P3)--(P4)--cycle;
    \draw[black!65] (P2)--(P3)--(P5)--cycle;
    \draw[black!65] (P2)--(P4)--(P5)--cycle;
    \draw[black!65] (P3)--(P4)--(P5)--cycle;

    \fill[red!25, opacity=.35] (P1)--(P3)--(P5)--cycle;
    \draw[red!70!black]         (P1)--(P3)--(P5)--cycle;

    \foreach \P in {P1,P2,P3,P4,P5}{\fill (\P) circle (0.015);}
  \end{scope}

\end{tikzpicture}
\caption{Optimal point placements for $n=3$ to $n=5$ for the Heilbronn triangle problem.}
\label{fig:n3_to_n5_placements}
\end{figure}
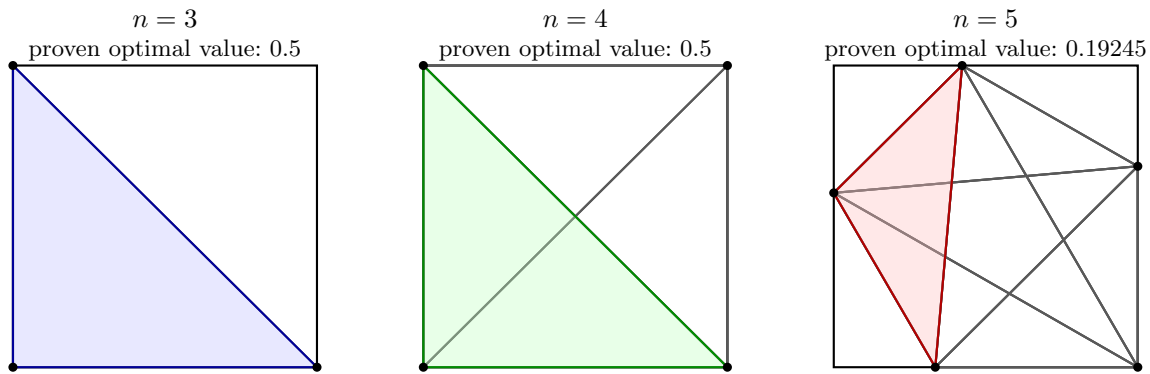



The study of Heilbronn's triangle problem has evolved through several distinct lines of inquiry, from asymptotic bounds to exact computational proofs. Early work focused on worst-case analytic bounds, a field initiated by Roth in 1951, who developed a machinery for proving upper bounds on $H_n^*$ and 
established that $H_n^* \ll \frac{1}{n\sqrt{\log\log n}}$ \cite{roth1951problem} (here, $A\ll B$ means that $A\le cB$ for an absolute constant $c$). He later sharpened this result by obtaining   bounds of the form $H_n^* \ll n^{-\mu+\varepsilon}$ for any $\varepsilon > 0$, where the exponent $\mu$ was improved from approximately 1.1056 \cite{roth1972problem, roth19722problem} to 1.117 \cite{roth1976developments}. 
Building on this framework, the authors of  \cite{komlos1981heilbronn} improved this bound, effectively increasing the exponent to $\mu = 8/7$, by proving $H_n^* \ll n^{-8/7+\varepsilon}$. 
Later,  the same authors established a complementary lower bound of $H_n^* \gtrsim (\log n)/n^{2}$ in~\cite{komlos1982lower}, disproving Heilbronn's original conjecture that $H_n^* = \Theta(1/n^{2})$. 


Complementing the asymptotic theory, researchers have also employed heuristic and constructive methods to find specific point placements that yield strong lower bounds~\cite{goldberg1972maximizing,comellas2002new,bertram2000algorithm}. Distinct from these heuristic approaches, exact results for specific small values of $n$ have been proven using analytical and computer-algebra methods. In particular, the precise values for $n=5$ and $n=6$ were established in \cite{lu1991goldberg}, and more recently for $n=7$ and $n=8$, proofs leveraging symbolic computation and automated deduction were used to certify optimal configurations~\cite{zeng2008heilbronn, dehbi2022heilbronn}. Note that the deterministic worst-case value $H_n$ contrasts sharply with the average case for random points, where the smallest triangle area is expected to be much smaller, on the order of $\Theta(1/n^{3})$ \cite{jiang1999expected, jiang2002average}.


For instances such as $n=9$, where exact analytical proofs have remained elusive, high-throughput discretization and parallel enumeration have produced extremely tight numerical bounds. For instance, using a specialized grid search method accelerated by GPGPU computing, the authors of \cite{chen2017searching} obtained the    following bounds:  $0.054875999 < H_{9}^* < 0.054878314$. They established this dual bound by employing a recursive branch-and-bound algorithm over increasingly fine grid discretizations of the unit square, where local upper bounds were calculated for each rectangle configuration to rigorously prune the search space~\cite{chen2017searching}.
In contrast, the approach adopted in this work frames the problem as a maximin mathematical program, leveraging deterministic global optimization techniques—specifically Mixed-Integer Nonlinear Programming (MINLP)—to achieve certified results. Rather than relying on exhaustive enumeration, our methodology constructs   optimality certificates by integrating tight continuous relaxations with adaptive discretization strategies. This distinction is crucial since, while analytical proofs for $n \ge 9$ remain elusive, existing computational methods often focus on heuristic lower bounds without providing the matching  upper bounds necessary for certification. Consequently, rigorous global optimization frameworks specifically tailored for certifying Heilbronn numbers have remained largely underdeveloped in literature.

In this paper, we apply global optimization methods to the Heilbronn triangle problem. 
In particular, we develop three novel mathematical programming models to solve this problem with global optimization solvers   Gurobi \cite{Gurobi} and BARON \cite{baron}. The difference between these models lies in the way they handle the non-convex relations in problem~\eqref{eq:heilbronnGeneric} as detailed in Section~\ref{Approaches}.
Two of these models belong to the class of mixed-integer quadratically constrained programming (MIQCP) whereas the other model belongs to the class of QCP. We propose several enhancements to strengthen these formulations by exploiting problem structure and using well-known techniques from global optimization literature such as bound tightening and symmetry breaking. We carry out extensive computations to decide the best performing model as well as the best performing solver in addition to the most effective {combination of} enhancements. Our rigorous analysis has enabled us to certify the exact value $H_9^* = 0.0548767$ (to within solver tolerances) for the first time in the literature, closing the optimality gap that remained in previous computational studies.
We now list our main contributions  in this paper:
\begin{itemize}

    \item \textit{Novel formulations.}
    We develop three novel formulations for  the Heilbronn triangle problem that can be used to solve    any instance of size \(n\). These formulations belong to the class of MIQCP or QCP, hence, any global optimization solver supporting these classes of problems is readily applicable to solve the Heilbronn triangle problem.

    \item \textit{Formulation enhancements.}
    We propose several enhancements to   strengthen  the proposed formulations. Some of these enhancements are \textit{universal}, in the sense that they are applicable to any instance with an arbitrary number of points $n$. Such enhancements involve bound tightening and symmetry breaking inequalities, commonly used in the global optimization literature. We also propose instance-specific enhancements that work for a particular value of $n$. For example, we show that for the nine-point problem, eight of these points must be located close to the boundary of the unit box (see  Proposition~\ref{prop:nine-eight-edge} for a precise statement).

    \item \textit{Certified solutions by instance size.}
    Using our novel formulations and utilizing our enhancements, we are able to efficiently recover the proven global optima for \(n=3,\dots,8\) within solver tolerances in the order of seconds. 
    We then certify the global optimum for \(n=9\) as $H_9^*=0.0548767$ by establishing a dual bound that matches the smallest area of a feasible configuration. This is a significant improvement over the previous state-of-the-art \cite{chen2017searching}, which employed a specialized GPGPU-based grid search over 31 days to establish a narrow bracket ($0.054875999 < H_{9}^* < 0.054878314$), whereas our method provides a certified value within one-day of computational effort.  
    Finally, for \(n=10\), we also provide, without an optimality certificate, a configuration that matches the best-known result. Our model found a placement yielding a lower bound of $H_{10}^* \ge 0.0465369$ in approximately 800 seconds, which is consistent with the best empirical value of $0.046537$ previously reported \cite{comellas2002new}.

\end{itemize}




The remainder of this paper is organized as follows. Section~\ref{sec:problem_and_approaches} formally introduces our three novel formulations for  the Heilbronn triangle problem. Section~\ref{sec:enhancements} proposes three groups of enhancements that significantly strengthen these formulations. Section~\ref{sec:computational_experiments} reports the results of our extensive computational study, which involves comparing   formulations and solvers together with an assessment of  the impact of each enhancement group. Finally, Section~\ref{sec:conclusions} summarizes our main findings and discusses directions for future work. 

\section{Problem Formulation and Solution Approaches}\label{sec:problem_and_approaches}

In this section, we formally define the Heilbronn triangle problem and present three novel mathematical programming formulations that can be solved using global optimization solvers. 

\subsection{Problem Formulation}
We first reformulate problem~\eqref{eq:heilbronnGeneric} as
\begin{align}
    \max_{x,y,H,S} \ & H_n \\
    \textup{s.t.} \  & |S_{ijk}| \ge H_n  &1&\le i < j < k \le n  \label{eq:S-and-H} \\
      \ &      S_{ijk} = \frac12 [ x_i(y_j-y_k) - x_j(y_i-y_k)+x_k(y_i-y_j) ] & 1&\le i<j<k \le n \label{eq:signedArea} \\
    &-\frac12 \le S_{ijk} \le \frac12  & 1 & \le i<j<k \le n  \label{eq:bounds-S}\\
     &     0 \le x_i, y_i \le 1  & 1& \le i \le n     \label{eq:bounds-xy} \\ 
    &      \underline H_n \le H_n \le \overline H_n , \label{eq:bounds-H} 
\end{align}
where variable $S_{ijk}$ represents the  \textit{signed area} of the triangle with corner points $i$, $j$ and $k$, and $H_n$ is the objective function variable representing the area of the smallest triangle of the Heilbronn triangle problem with $n$ points. Note that bounds on variable $S_{ijk}$ stated in constraint~\eqref{eq:bounds-S} follow from the definition of this variable given in constraint~\eqref{eq:signedArea} and the fact that   $0 \le x_i, y_i \le 1$ as given in constraint~\eqref{eq:bounds-S}. The procedure to obtain bounds on variable~$H_n$ is explained later in Section~\ref{sec:boundsHn}.
While constraint~\eqref{eq:bounds-H} might appear redundant for the problem definition, these \emph{a priori} bounds are computationally critical. In particular, the upper bound $\overline{H}_n$ is essential for constructing a tight linearization of the absolute value constraint~\eqref{eq:S-and-H} (as detailed in \Cref{Approach1}), and the lower bound $\underline{H}_n$ allows the solver to effectively prune the branch-and-bound search tree.

Notice that the above formulation is highly non-linear and non-convex due to the reverse convex constraint~\eqref{eq:S-and-H} and the bilinear constraint~\eqref{eq:signedArea}. In the sequel, we propose three novel approaches to obtain MIQCP or QCP formulations so that available global optimization  solvers can be directly used.

\subsection{Solution Approaches}\label{Approaches}
The primary challenge in formulating the Heilbronn triangle problem for global optimization solvers is handling the non-convex absolute value function required for the standard area calculation. In this section, we present three novel mathematical programming approaches, each designed to address this core difficulty in a different way. 

\subsubsection{Approach 1}
\label{Approach1}

The main idea of this approach is to rewrite   the non-convex constraint $H_n \le |S_{ijk}|$ as a pair of linear inequalities with the help of additional binary variables.  To achieve this, a binary variable $b_{ijk}$ is introduced for each triangle to disjunctively model the sign of the area $S_{ijk}$. 
In particular, the desired relationship is established using the following   set  of inequalities:
\begin{subequations}\label{eq:linearWithBinary}
\begin{align}
  &  (1-b_{ijk})\left(\overline{H}_n+\frac{1}{2}\right)  + S_{ijk} \ge H_n     &1&\le i<j<k \le n \\
  &  \ b_{ijk}\left(\overline{H}_n+\frac{1}{2}\right)  - S_{ijk} \ge H_n     &1&\le i<j<k \le n \\
 & -\frac12 \le S_{ijk} - \left(\frac12 + \underline H_n\right) b_{ijk} \le -\underline H_n   &1& \le i<j<k \le n \\
    &   \ b_{ijk}\in\{0,1\}   &1&\le i<j<k \le n .
\end{align}
\end{subequations}
The first pair of inequalities in \eqref{eq:linearWithBinary} encodes the non-convex constraint $H_n \le |S_{ijk}|$ via  binary variables while the second pair is also valid and is included in the implementation to strengthen the formulation.

Then, we obtain the following equivalent mixed-integer quadratically constrained program (MIQCP):
\[
\max_{x,y,H, S, b}   \left\{ H_n : \eqref{eq:signedArea}-\eqref{eq:linearWithBinary}   \right\}.
\]


To further structure the model, we isolate the bilinear terms from the area formula by defining new variables $w_{ij}$:
\begin{equation}\label{eq:w=xy}
w_{ij} = x_i y_j \quad 1\le i, j \le n,
\end{equation}


This results in a linear expression for the signed area in terms of these new variables; and rewriting~\eqref{eq:signedArea} as
\begin{equation}\label{eq:signedArea-w}
    S_{ijk} = \frac12 [ (w_{ij}-w_{ik}) - (w_{ji}-w_{jk}) + (w_{ki}-w_{kj}) ] 
    \quad 1\le i<j<k \le n.
\end{equation}
The resulting MIQCP formulation, which is the form we use in our implementation, is as follows:
\[
\max_{x,y,w,H, S, b}   \left\{ H : \eqref{eq:bounds-S}- 
\eqref{eq:signedArea-w}  \right\}.
\]
This model is an MIQCP because it contains both the binary variables $b_{ijk}$ and the non-convex quadratic equality constraints in \eqref{eq:w=xy}.






\subsubsection{Approach 2}
\label{Approach2}

The main idea of this approach is to avoid the binary variables of \Cref{Approach1} by using a non-convex quadratic formulation. Instead of linearizing the absolute value, we introduce a new continuous, non-negative variable, $U_{ijk} \ge 0$, to represent the (unsigned) area of the triangle formed by points $i$, $j$, and~$k$.
The non-convex constraint $H_n \le |S_{ijk}|$ from \eqref{eq:S-and-H} is then replaced by the following pair of constraints:
\begin{align}
    U_{ijk} &\ge H_n  \ & 1&\le i<j<k \le n  \label{eq:U_geq_H}  \\
    S_{ijk}^2 &= U_{ijk}^2  & 1&\le i<j<k \le n. \label{eq:S_sq_eq_U_sq}
\end{align}
Since $U_{ijk} \ge 0$ is implied by the fact that $H_n \ge 0$, the quadratic equality \eqref{eq:S_sq_eq_U_sq} is equivalent to $U_{ijk} = |S_{ijk}|$. This substitution transforms the model into a  QCP  that contains no binary variables.

As in Approach 1, we use the intermediate $w_{ij}$ variables defined in \eqref{eq:w=xy} and the linear expression for $S_{ijk}$ from \eqref{eq:signedArea-w}. The resulting QCP formulation, which matches our  implementation, is as follows:
\[
\max_{x,y,w,H, S, U} \left\{ H_n : \eqref{eq:bounds-S}-\eqref{eq:bounds-H}, \eqref{eq:w=xy}, \eqref{eq:signedArea-w}, \eqref{eq:U_geq_H}, \eqref{eq:S_sq_eq_U_sq}  \right\}.
\]
Notice that this formulation is a non-convex QCP due to quadratic equality constraints in~\eqref{eq:w=xy} and~\eqref{eq:S_sq_eq_U_sq}.



\subsubsection{Approach 3}
\label{Approach3}

The main idea of this approach is to discretize the continuous variable  $x_i \in [0, 1]$ in order to obtain a formulation with a strong relaxation. In particular, for a positive integer $P$, 
we discretize variable $x_i$  as
\[
x_i = \sum_{p=1}^P 2^{-p} \xi_{ip} + \epsilon_i, \quad \text{where } \xi_{ip} \in \{0,1\} \text{ and } \epsilon_i \in [0, 2^{-P}].
\]
Then, we have
\begin{equation*}
    w_{ij} = x_i y_j = \sum_{p=1}^P 2^{-p} \xi_{ip} y_j + \epsilon_i y_j, 
    \quad 1 \le i,j \le n.
\end{equation*}
Observe that we  rewrite the `continuous-times-continuous' term  $w_{ij} = x_i y_j$, as a sum of `binary-times-continuous' products $\xi_{ip} y_j$ and another `continuous-times-continuous' term  $\epsilon_i y_{j}$, which involves a non-negative variable $\epsilon_i$ with a small upper bound. This new structure is highly advantageous because `binary-times-continuous' products can be linearized exactly using standard McCormick envelopes defined in Fact~\ref{fact:McCormick} and `continuous-times-continuous' product is easier to handle  when at least one variable has a small range. 
\begin{fact}\label{fact:McCormick}
    The convex hull of the set $\{ (x,y) \in [\underline x, \overline x]\times[\underline y, \overline y]: w = xy\}$ is given by the McCormick envelopes~\cite{McCormick1976}:
    \[
    \mathcal{M}(\underline x, \overline x; \underline y, \overline y) = \{ (x,y,w) : \
    w \ge \underline x y + x \underline y  -\underline x \underline y , \
    w \ge \overline x y+x \overline y - \overline x \overline y , \
    w \le \overline x y+x \underline y - \overline x \underline y , \
    w \le \underline x y+x \overline y  - \underline x \overline y \} .
    \]
\end{fact}
%
This reformulation is expressed through the following set of constraints, where new variables $\phi_{ipj}$ and $\omega_{ij}$ are introduced to represent the product terms:
\begin{subequations}\label{eq:discretization-Exact}
     \begin{align}
     &   w_{ij}=\sum_{p=1}^P 2^{-p} \phi_{ipj} + \omega_{ij} &1& \le i,j \le n \label{eq:discretization-Exact1} \\
    & (\xi_{ip}, y_j, \phi_{ipj} ) \in \mathcal{M}(0,1; 0,1) &1& \le i, j \le n, \ p=1,\dots,P \label{eq:discretization-Exact2} \\
    & \xi_{ip} \in \{0,1\}  &1& \le i \le n, \ p=1,\dots,P \label{eq:discretization-Exact2-} \\
    &
    \omega_{ij} = \epsilon_i y_j &1& \le i,j \le n \label{eq:discretization-Exact3}\\
    &
    0 \le \epsilon_i \le 2^{-P} &1& \le i  \le n. \label{eq:discretization-Exact4}
     \end{align}
\end{subequations}
The resulting MIQCP formulation is as follows:
\[
\max_{x,y,w,H, S, b, \xi, \phi, \omega}   \left\{ H : \eqref{eq:bounds-S}-\eqref{eq:linearWithBinary} , 
\eqref{eq:signedArea-w}, \eqref{eq:discretization-Exact}  \right\}.
\]
While this formulation is exact, it still contains the non-convex quadratic constraint \eqref{eq:discretization-Exact3}. To obtain a model solvable with MILP techniques, we relax this final non-convex term. This gives rise to a Mixed-Integer Linear Programming (MILP) relaxation, which we use in our computational study in \Cref{ComparingApproaches}. This MILP relaxation is formulated as follows: 
\[
\max_{x,y,w,H, S, b, \xi, \phi, \omega}   \left\{ H : \eqref{eq:bounds-S}-\eqref{eq:linearWithBinary} , 
\eqref{eq:signedArea-w}, \eqref{eq:discretization-Exact1}- \eqref{eq:discretization-Exact2-} ,
(\epsilon_i, y_j, \omega_{ij} ) \in \mathcal{M}(0, 2^{-P}; 0,1 ) \ 1 \le i,j \le n  \right\}.
\]

This approach follows a line of papers  that has successfully implemented similar ideas   on related problems, including the pooling problem   \cite{dey2015analysis,dey2020convexifications, jalilian2023improved} and the circle packing problem \cite{MaxNumber_Litvinchev2015_general,tacspinar2024discretization}, among others.

Note that all the formulations introduced in this section are challenging to solve in practice. In order to accelerate the solution process, we introduce several enhancements  in the next section.

\section{Enhancements}\label{sec:enhancements}

We introduce several enhancements 
to strengthen the formulations introduced before and shorten the overall computation time. 
We categorize these enhancements into thematic categories, hereafter termed as \textit{enhancement groups}, and assess the impact of each group as a whole rather than of every individual enhancement as detailed below. 

\subsection{First Group: Bound Tightening and Symmetry Breaking}


\subsubsection{Bounds for the Objective Value \(H_n\)}
\label{sec:boundsHn}
       
Note that  any feasible solution yields a valid lower bound for a maximization problem and it is trivial to generate feasible \(n\)-point configurations in our case. To this end, we sample  \(10^{6}\) i.i.d.  \(n\)-point configurations from the unit square $[0,1]^2$, and we set the lower bound for $H_n$ as
        \[
          \underline H_n \;=\; \max\{\text{minimum-triangle-area of the sample}\}.
        \]

To obtain an upper bound on $H_n$, we rely on the value of $H_{n-1}^*$ as formalized below:

\begin{prop}\label{prop:n-1PointGivesUpperBound}
For \(n>3\), the optimal minimum‑triangle area is non-increasing in \(n\); in particular,
\[
  H_n^* \;\le\; H_{n-1}^* .
\]
\end{prop}
\begin{proof}
Consider an optimal solution of the $n$-point problem given as $(x_i^*,y_i^*)$, $i=1,\dots,n$ with the minimum area of $H_n^*$. Choose a point $(x_j^*,y_j^*)$ among these $n$ points such that the  area of the smallest triangle formed by the remaining $n-1$ points is still $H_n^*$ (note that such a point exists: either all triangles formed by $n$ points have area of $H_n^*$, in which case any point can be chosen; or there exists a point which is not a corner of the smallest triangle that can be chosen). Clearly, the points $(x_i^*,y_i^*)$, $i=1,\dots,n, i \neq j$ is a feasible solution of the   $(n-1)$-point problem with objective function value of $H_n^*$. Therefore, we conclude that $H_n^* \le H_{n-1}^*$.
\end{proof}
Due to Proposition~\ref{prop:n-1PointGivesUpperBound}, we set $\overline{H}_n = H_{n-1}^*$.

\subsubsection{Symmetry Breaking}\label{sec:order-conventions}

Similar to other geometry optimization problems, the Heilbronn triangle problem has many symmetrical solutions in the sense that by relabeling the points corresponding to a certain solution, we can obtain solutions that are fundamentally the same. In order to prevent this issue, which causes the optimization solvers to stall, we propose several symmetry breaking constraints given in the next proposition.

\begin{prop}\label{prop:order-conventions}
There exists an optimal solution  $\{(x_k,y_k)\}_{k=1}^n$ to  the Heilbronn problem satisfying the following inequalities: 
\begin{enumerate}[(i)]
  \item \(
     y_1 \;\le\; y_2 \;\le\; \cdots \;\le\; y_n  
  \).

  \item 
  \( w_{i1} \le w_{i2} \le \cdots \le w_{in} \). 

\item \(
  x_2 \;\ge\; x_1  \text{ and } x_1 \;\le\; \tfrac{1}{2}
\).
\end{enumerate}
\end{prop}
\begin{proof}
In the proof, we will use the fact  that the feasible set and objective function are invariant under point relabelings  and under the vertical reflection \(x\mapsto 1-x\).  
(i) This result follows by sorting the  set \(\{y_k\}_{k=1}^n\) and relabeling accordingly to obtain \(y_1\le\cdots\le y_n\).
(ii) This result is a consequence of Item (i) and the definition of $w_{ij}=x_iy_j$.
(iii) 
This result is verified as follows: If \(x_2<x_1\), then we swap labels \(1\) and \(2\).
If   \(x_1>\tfrac12\), we reflect the entire configuration across the line \(x=\tfrac12\) via \((x,y)\mapsto (1-x,y)\) (and, if needed, re-swap labels \(1\) and \(2\) to keep \(x_2\ge x_1\)).
Both relabeling and reflection preserve all triangle areas and the constraints of the square, hence feasibility and the objective value are preserved.
\end{proof}

\subsection{Second Group: Points on the Boundary} \label{sec:BoundTightening}

As the Heilbronn triangle problem seeks to locate points in  a way that the area of the smallest triangle is maximized, it is intuitive to expect some points to be located at the edges of the unit square $[0,1]^2$. The second group of enhancements formalizes this intuition.

\subsubsection{At Least One Point on Each Edge (\(n \ge 4\))}\label{sec:onepointeachedge}


\begin{prop}\label{onepointoneachedge}
For any optimal solution of the Heilbronn triangle problem with at least four points, each edge of the unit square $[0,1]^2$ contains at least one point. 
\end{prop}
\begin{proof}
Let \(a=\min_i x_i\), \(b=\max_i x_i\), \(c=\min_i y_i\), and \(d=\max_i y_i\).
If \([a,b]\times[c,d]\) is a \emph{proper} subrectangle of \([0,1]^2\), then \((b-a)(d-c)<1\).
Consider the affine map \(T:\,[a,b]\times[c,d]\to[0,1]^2\),
\[
T(x,y)=\Big(\tfrac{x-a}{\,b-a\,}, \tfrac{y-c}{\,d-c\,}\Big).
\]
For any triangle \(\Delta\) determined by three of the points, \(\operatorname{area}(T(\Delta))=\operatorname{area}(\Delta)/((b-a)(d-c))\).
Since \((b-a)(d-c)<1\), the area of each triangle   strictly increases under \(T\), so the minimum triangle area strictly increases as well—contradicting optimality.
Hence, we obtain \((b-a)(d-c)=1\), which implies that \(a=0\), \(b=1\), \(c=0\), and \(d=1\).
\end{proof}
The following corollary is a consequence of Proposition~\ref{onepointoneachedge}.

\begin{cor}\label{cor:implementingx}
For the Heilbronn triangle problem with at least four points, there exists an optimal solution that satisfies the following constraints:
\begin{enumerate}[(i)]
    \item $y_1 = 0$ and $y_n = 1$.
    \item $w_{i1} = 0$ and $w_{in} = x_i$ for $i=1,\dots,n$.
    \item For some binary  variables $c_{1i},c_{2i}\in\{0,1\}$,  $i=1,\dots,n$, 
    \begin{subequations}
        \begin{align}
        &   1 \le \sum_{i=1}^{n} c_{1i} \le 2,  \ \  x_i \leq 1 - c_{1i}  , \ i=1,\dots,n  \label{eq:x=0-onepoint} \\
        &     1 \le \sum_{i=1}^{n} c_{2i} \le 2 \label{eq:x=1-onepoint} , \ \  x_i \geq c_{2i} , \ i =1,\dots,n  .
        \end{align}
    \end{subequations}  
\end{enumerate}
\end{cor}

\begin{proof}
    Recall that each edge of the unit square $[0,1]^2$ must contain at least one point due to Proposition~\ref{onepointoneachedge}. Notice that equations in Item (i) follow  from Proposition~\ref{prop:order-conventions}(i) and guarantee that   edges $y=0$ and $y=1$ of the box $[0,1]^2$ contain at least one point.
    Similarly, equations in Item (ii) follow from Proposition~\ref{prop:order-conventions}(ii). 

    Let us now look at equations in Item (iii). In fact, inequalities~\eqref{eq:x=0-onepoint}  (resp. inequalities~\eqref{eq:x=1-onepoint}) guarantee that at least one point is located at the edge $x=0$ (resp. $x=1$) of the box $[0,1]^2$. Moreover, the upper bound of two makes sure that no more than two points are located in these edges, since otherwise, there would be three collinear points.
\end{proof}

\subsubsection{Two Points on Some Edge \;(\(n\in\{7,8\}\))}\label{2pointsedge}

For some \(n\) values, optimal configurations place multiple points on the boundary of the unit square. In particular, for 7-point and 8-point problems, the following result is proven in references \cite{zeng2008heilbronn}  and \cite{dehbi2022heilbronn}, respectively.
\begin{prop}\label{twopointsononeedge}
For any optimal solution of the Heilbronn triangle problem with seven or eight points, there exists an edge of the unit square $[0,1]^2$ that contains exactly two points.
\end{prop}

\begin{rem}
While this property is proven in the literature for $n=7$ \cite{zeng2008heilbronn} and $n=8$ \cite{dehbi2022heilbronn}, a corresponding mathematical proof for $n=9$ is not known. As detailed in \Cref{sec:nine-eight-edge}, we are able to computationally certify a related 
near-boundary property for $n=9$.
\end{rem}




The following corollary is a consequence of Proposition~\ref{prop:order-conventions} and  Proposition~\ref{twopointsononeedge}.
\begin{cor}\label{cor:implementingx78}
For the Heilbronn triangle problem with seven or eight points, there exists an optimal solution that satisfies the following constraints:
\begin{enumerate}[(i)]
    \item $y_1 = y_2 = 0$.
    \item $w_{i1} =w_{i2} = 0$ for $i=1,\dots,n$.
\end{enumerate}
\end{cor}




\subsubsection{Approximately Eight Boundary Points (\(n=9\))}
\label{sec:nine-eight-edge}

In all of our experiments related to the Heilbronn triangle problem with at least nine points,  we observe that eight points are placed  on the boundary of the unit square $[0,1]^2$ across all best configurations. The same pattern was also seen in all best-known configurations for \(n\ge 9\) reported in the literature. Hence, we have come up with the following conjecture:
\begin{conjecture}\label{conj:atleast2ptsonedges}
    For any optimal solution of the Heilbronn triangle problem with at least nine points, each edge of the unit square $[0,1]^2$ contains exactly two points. 
\end{conjecture}
Unfortunately, we are unable to give a mathematical proof of this conjecture. Instead, using Gurobi, we prove   a slightly weaker result in Proposition~\ref{prop:nine-eight-edge}, which states that eight points are approximately located near edges in an optimal solution. Note that this result is a strengthening over Proposition~\ref{onepointoneachedge}. 


\begin{prop}\label{prop:nine-eight-edge}

%
Let \(\varepsilon=10^{-2}\). 
For the Heilbronn triangle problem with nine points, there exists an optimal solution that satisfies the following constraints:
\begin{enumerate}[(i)]
    \item $y_1 = 0$ and $y_2 \le \varepsilon $.
    \item  $y_8 \ge 1-\varepsilon $ and $y_9 = 1$.
    \item For some binary  variables $c_{1i},c_{2i}\in\{0,1\}$,  $i=1,\dots,n$, 
    \begin{subequations}
        \begin{align}
        &   \sum_{i=1}^{n} c_{1i} = 2, \  x_i \leq 1+\varepsilon - c_{1i} , \ i =1,\dots,n \label{eq:x=0-points9}\\
        &     \sum_{i=1}^{n} c_{2i} = 2 , \  x_i \geq c_{2i}-\varepsilon , \ i =1,\dots,n  \label{eq:x=1-points9} .
        \end{align}
    \end{subequations}  
\end{enumerate}
\end{prop}

\begin{proof}
We compare two types of solutions as below: 


\textbf{Solution A: Exactly eight points near edges.}
Let us solve a restricted version of the Heilbronn triangle problem with nine points in which we assume that Conjecture~\ref{conj:atleast2ptsonedges} holds true. In particular, we enforce the following restrictions:  
\[
y_1=y_2=0, \ y_8=y_9=1, \ \sum_{i=1}^{9} c_{1i}= 
\sum_{i=1}^{9} c_{2i} = 2, 
c_{2i} \le x_i \le 1 - c_{1i}
,\ i=1,\dots,9.
\]
Using Approach 1, augmented with   other enhancements, Gurobi finds a feasible solution with an objective value of
\(L_A=\underline H_9=\texttt{0.0548765299}\) in \(317.611\) seconds. 

\textbf{Solution B: At most seven points near edges.}
Let us solve another restricted version of the Heilbronn triangle problem with nine points in which we assume at most seven points can be near edges, which is guaranteed that there exists a point whose distance to all edges is at least $\varepsilon$. In particular, we enforce the following restrictions for this purpose, $i=1,\dots,n$ with $n=9$:


\[
\begin{aligned}
&y_1 = 0   && \text{(at least  one point  on  }  y=0) \\
&y_9 = 1 && \text{(at least  one point  on  }  y=1)  \\[6pt]
&x_i \;\le\; 1 - c_{1i} &
            &\text{(forces }x_i=0\text{ if }c_{1i}=1) \\[4pt]
&x_i \;\ge\; c_{2i} &
            &\text{(forces }x_i=1\text{ if }c_{2i}=1) \\[4pt]
&x_i \;\ge\; \varepsilon\,(1 - c_{1i}) &
            &\text{(otherwise }x_i \ge \varepsilon\text{ away from }x=0) \\[4pt]
&x_i \;\le\; 1 - \varepsilon\,(1 - c_{2i}) &
            &\text{(otherwise }x_i \le 1-\varepsilon\text{ away from }x=1) \\[8pt]
&1 \;\le\; \sum_{i=1}^{n} c_{1i} \;\le\; 2 &
            &\text{(one or two points within }\varepsilon\text{ of }x=0) \\[4pt]
&\sum_{i=1}^{n} c_{2i} \;=\; 1. &
            &\text{(exactly one point within }\varepsilon\text{ of }x=1)
\end{aligned}
\]

Hence, at most two points may lie within $\varepsilon$ of each of the edges $y=0$, $y=1$, and $x=0$, while at most one point may lie within $\varepsilon$ of the edge $x=1$, giving a maximum of $2+2+2+1=7$ near-boundary points in total.

Using Approach 1, adapted with the Solution B constraints and all other proven enhancements,
Gurobi is able to provide a certified upper bound of
\(
U_B \;=\; \overline H_9 \;=\; \texttt{0.0546346906} 
\)
in \(82\ 359.163\) seconds.

Notice that the
gap between two bounds $L_A$ and $U_B$ is computed as 
\[
\lvert L_A-U_B\rvert
=0.0548765299-0.0546346906
=0.0002418393 \;>\; 0.
\]
Hence, we deduce that a globally optimum  solution of the Heilbronn triangle problem with nine points cannot be of type B. Hence, it must be of type A. 

Note that no edge can host three points within distance \(\varepsilon\): such a triple would form a triangle of area at most \(\varepsilon/2 = 5\times10^{-3}\), far below the certified lower bounds, so the model remains valid while capping the total number of near-boundary points at seven.
\end{proof}
Notice that  equations in Item (i) (resp. Item (ii)) imply that two points must be near the edge $y=0$ (resp. $y=1$) while equations in Item (iii) imply that two points must be near edge $x=0$ (equations~\eqref{eq:x=0-points9} and $x=1$ (equations~\eqref{eq:x=1-points9}) each.

\begin{rem}
In the proof of Proposition~\ref{prop:nine-eight-edge}, we also tried   a lighter alternative for Solution B in which we replaced the horizontal near-edge logic by the single inequality \(y_{2}\ge \varepsilon\). This enforces exactly one point on \(y=0\) (since \(y_1=0\)) while keeping the same vertical near-edge controls as above, and it permits up to two points near each vertical edge and one point on \(y=1\). Although this variant is syntactically simpler, in our experiments it was {less time-efficient} than the Scenario~B formulation stated above.
\end{rem}

\subsection{Third Group: Local Packing and Separation }

\subsubsection{Rectangles with at Most Two Points} \label{rec2point}

The following enhancement is based on a simple geometric observation about
feasible configurations of the Heilbronn triangle problem: In any optimal solution, a rectangle whose area is less than $2\underline{H}_n$ cannot contain three points, where $\underline{H}_n \le H_n^*$. 
We first
formalize this observation in Proposition~\ref{prop:two-per-rectangle-general}, 
and then show how we exploit it in
our formulation via a particular family of rectangles in Corollary~\ref{cor:two-per-rectangle-encoding}.

\begin{prop}\label{prop:two-per-rectangle-general}
Consider   $\underline{H}_n \le H_n^*$ and
let   $R \subset [0,1]^2$ be a rectangle whose area is less than $2\underline{H}_n$.
Then, in any optimal solution of the    
Heilbronn triangle problem with $n$ points, 
 $R$ contains at most two points.
\end{prop}
\begin{proof}
Let $(x^*,y^*)\in[0,1]^n\times[0,1]^n$ be an  optimal solution of the    
Heilbronn triangle problem with $n$ points. 
Assume, for contradiction, that $R$ contains  points $i$, $j$ and $k$. 
Then, the area of the triangle formed by these points is at most $\tfrac12\,\mathrm{area}(R) < \underline{H}_n$. However, this contradicts the optimality of  $(x^*,y^*)$ since the area of this triangle must be at least $ \underline{H}_n$.
\end{proof}
%




We  now explain how we utilize Proposition~\ref{prop:two-per-rectangle-general} to strengthen our formulations for the Heilbronn triangle problem. In particular, we  partition   the unit square into horizontal strips whose height is less than   $2\underline{H}_n$. Then, by the introduction of a new set of binary variables $r_{ij}$, we enforce that i) each strip can contain at most two points and ii) each point must be contained in exactly one strip. \Cref{fig:two_per_rectangle} illustrates this partitioning strategy and Corollary~\ref{cor:two-per-rectangle-encoding} formalizes this argument.

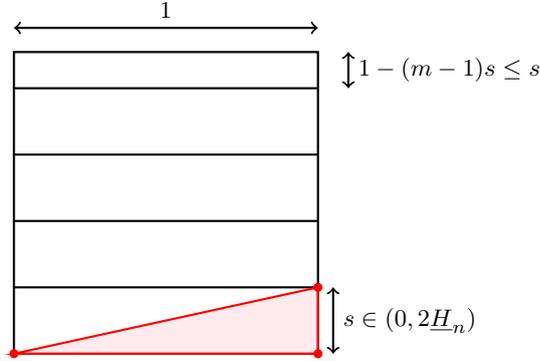
\begin{figure}[h!]
\centering
\begin{tikzpicture}[scale=4, thick]

  \def\s{0.22}
  \pgfmathsetmacro{\yone}{\s}
  \pgfmathsetmacro{\ytwo}{2*\s}
  \pgfmathsetmacro{\ythree}{3*\s}
  \pgfmathsetmacro{\yfour}{4*\s}
  \pgfmathsetmacro{\topH}{1-\yfour}

  \draw (0,0) rectangle (1,1);
  \draw (0,\yone)   -- (1,\yone);
  \draw (0,\ytwo)   -- (1,\ytwo);
  \draw (0,\ythree) -- (1,\ythree);
  \draw (0,\yfour)  -- (1,\yfour);

  \coordinate (Pleft)  at (0, 0);
  \coordinate (PbotR)  at (1, 0);
  \coordinate (PtopR)  at (1, \s);

  \fill[red!25,opacity=.30] (Pleft) -- (PbotR) -- (PtopR) -- cycle;
  \draw[red]                (Pleft) -- (PbotR) -- (PtopR) -- cycle;

  \foreach \pt in {Pleft,PbotR,PtopR}{
      \fill[red] (\pt) circle (0.015);
  }

  \draw[<->] (0,1.08) -- (1,1.08);
  \node[above, font=\scriptsize] at (0.5,1.08) {$1$};

  \draw[<->] (1.05,0) -- (1.05,\s);
  \node[right, font=\scriptsize]
        at (1.05,0.5*\s)
        {$\displaystyle s \in (0,2\underline{H}_n)$};

  \draw[<->] (1.10,\yfour) -- (1.10,1);
  \node[right, font=\scriptsize]
        at (1.10,\yfour + 0.5*\topH)
        {$\displaystyle 1-(m-1)s \le s$};

\end{tikzpicture}
\vspace{0.4em}
\\
\caption{The unit square partitioned into four strips of height $s$ and a final strip of height $1-(m-1)s<s$. Each strip has area less than $2\underline{H}_n$, so it can contain at most two points.}
\label{fig:two_per_rectangle}
\end{figure}

\begin{cor}\label{cor:two-per-rectangle-encoding}
Consider   $\underline{H}_n \le H_n^*$ and 
   $  s \in (0,   2\underline{H}_n)$.
Let $m$ be the smallest integer such that $(m-1)s < 1 \le m s$ and partition the unit interval $[0,1]$ into $m$ subintervals
$[L_p, U_p]$ such that $L_p=(p-1)s$ and $U_p=\min\{1,ps\}$, for $p=1,\dots,m$.
Then, there exists an optimal solution to the Heilbronn triangle problem that satisfies the following constraints for some binary variables  $    r_{pi} \in \{0,1\}$, $p=1,\dots,m, i=1,\dots,n  $: 
\begin{enumerate}[(i)]
    \item $\sum_{i=1}^{n} r_{pi} \le 2,\quad p=1,\dots,m$.
    \item  $\sum_{i=1}^{m} r_{pi} = 1,\quad i=1,\dots,n$.
    \item 
    $  L_p - (1 - r_{pi}) \le 
y_i \le U_p + (1 - r_{pi}) \quad p=1,\dots,m, i=1,\dots,n  $.
\end{enumerate}

\end{cor}

In our computational experiments, we instantiate Corollary~\ref{cor:two-per-rectangle-encoding} with the specific
choice of 
\(
s \;=\; 2\underline{H}_n -  10^{-6}
\). 
This choice guarantees that the tolerance is  small enough to be negligible compared to
$\underline{H}_n$, but large enough to prevent round-off errors from making the strip
area equal to or slightly larger than $2\underline{H}_n$. 

\begin{rem}\label{Remark3}
Proposition~\ref{prop:two-per-rectangle-general} suggests that one could, in
principle, strengthen Corollary~\ref{cor:two-per-rectangle-encoding} by
introducing additional families of rectangles so as to cover more (or even all)
rectangles in $[0,1]^2$ whose area is below the threshold $2\underline{H}_n$.
For example, one natural extension would be to replicate the construction of
Corollary~\ref{cor:two-per-rectangle-encoding} with vertical strips, or to add
further rotated or non-uniform rectangles, each with its own set of binary
assignment variables and capacity constraints. However, when we implemented the
same strip-based encoding also for vertical rectangles, we quickly observed that
the additional binary variables substantially increased the model size and
solution time, while providing only marginal tightening in practice. On the
basis of these preliminary experiments, we therefore chose to restrict this
enhancement to the horizontal strips used in
Corollary~\ref{cor:two-per-rectangle-encoding}.
\end{rem}

\subsubsection{Small Squares Containing at Most One Point}

The second enhancement is based on  another geometric intuition: In any optimal solution, a square whose side is less than $\underline{H}_n$ cannot contain two points, where $\underline{H}_n \le H_n^*$ (otherwise, those
two points together with any third point in $[0,1]^2$ would form a triangle
whose area is necessarily below the given lower bound $\underline{H}_n$). 
We first
formalize this observation in Proposition~\ref{prop:one-per-square-general}, 
and then show how we exploit it in
our formulation via a particular family of squares in Corollary~\ref{cor:one-per-subsquare-encoding}.


\begin{prop}\label{prop:one-per-square-general}
Consider   $\underline{H}_n \le H_n^*$ and
let   $R \subset [0,1]^2$ be a square whose side is less than $ \underline{H}_n$.
Then, in any optimal solution of the    
Heilbronn triangle problem with $n$ points, 
 $R$ contains at most one point.
\end{prop}
\begin{proof}
Let $(x^*,y^*)\in[0,1]^n\times[0,1]^n$ be an  optimal solution of the    
Heilbronn triangle problem with $n$ points. 
Assume, for contradiction, that $R$ contains  points $i$ and and $j$. Observe that the distance between these points is   less than $\underline{H}_n \sqrt{2}$.
Now, consider a   point  $(\tilde x, \tilde y)\in[0,1]^2$. Notice that the area of triangle with corners $(x_i^*,y_i^*)$, $(x_j^*,y_j^*)$ and $(\tilde x, \tilde y)$ is less than $\frac12  (\underline{H}_n \sqrt{2}) \sqrt{2} =\underline{H}_n $ (since the perpendicular distance of the point $(\tilde x, \tilde y)$ to the line segment formed by points $(x_i^*,y_i^*)$ and $(x_j^*,y_j^*)$ is at most $\sqrt{2}$).   However, this contradicts the optimality of  $(x^*,y^*)$ since the area of this triangle must be at least $ \underline{H}_n$.
\end{proof}


We  again explain how we utilize Proposition~\ref{prop:two-per-rectangle-general} to strengthen our formulations for the Heilbronn triangle problem. In particular, we  partition   the unit square into a regular $m\times m$ grid of axis-aligned
subsquares of side length $1/m < \underline{H}_n$. Then, by the introduction of a new set of binary variables $u_{pqi}$, we enforce that i) each subsquare can contain at most one point  and ii) each point must be contained in exactly one subsquare. \Cref{fig:one_per_subsquare} illustrates this partitioning strategy and Corollary~\ref{cor:one-per-subsquare-encoding} formalizes this argument.

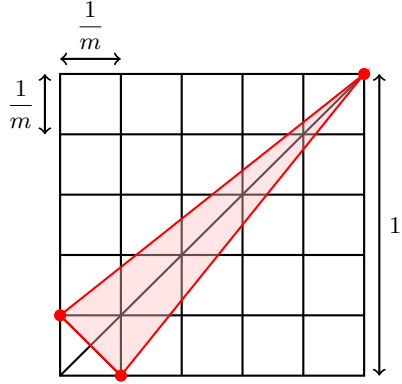
\begin{figure}[h!]
\centering
\begin{tikzpicture}[scale=4, thick]

  \def\N{5}                     
  \pgfmathsetmacro{\h}{1/\N}    

  \draw (0,0) rectangle (1,1);
  \foreach \k in {1,...,4}{
     \draw (0,\k*\h) -- (1,\k*\h);   
     \draw (\k*\h,0) -- (\k*\h,1);   
  }

  \draw (0,0) -- (1,1);        
  \draw (0,\h) -- (\h,0);      

  \coordinate (A) at (1,1);
  \coordinate (B) at (\h,0);
  \coordinate (C) at (0,\h);

  \begin{scope}
    \clip (-0.02,-0.02) rectangle (1.02,1.02);
    \fill[red!25,opacity=.40] (A)--(B)--(C)--cycle;
    \draw[red] (A)--(B)--(C)--cycle;
    \foreach \pt in {A,B,C}{\fill[red] (\pt) circle(0.02);}
  \end{scope}

  \draw[<->] (0,1.05) -- (\h,1.05);
  \node[above, font=\scriptsize] at (\h/2,1.05) {$\dfrac{1}{m}$};

  \draw[<->] (-0.05,1) -- (-0.05,1-\h);
  \node[left, font=\scriptsize]  at (-0.05,1-\h/2) {$\dfrac{1}{m}$};

  \draw[<->] (1.05,0) -- (1.05,1);
  \node[right, font=\scriptsize] at (1.05,0.5) {$1$};

\end{tikzpicture}
\caption{Partitioning the unit square into an $m \times m$ grid. If $1/m < \underline{H}$, then each subsquare can contain at most one point.}
\label{fig:one_per_subsquare}
\end{figure}

\begin{cor}\label{cor:one-per-subsquare-encoding}
Consider   $\underline{H}_n \le H_n^*$
and let $m\ge 1$ be an integer such that $1/m < \underline{H}_n$. 
Equipartition the unit interval $[0,1]$ into $m$ subintervals
$[L_p, U_p]$ such that $L_p=\frac{p-1}{m}$ and $U_p=\frac{p}{m}$, for $p=1,\dots,m$. 
Then, there exists an optimal solution to the Heilbronn triangle problem that satisfies the following constraints for some binary variables  $    u_{pqi} \in \{0,1\}$, $p=1,\dots,m,  q=1,\dots,m, i=1,\dots,n  $: 
\begin{enumerate}[(i)]
    \item $\sum_{i=1}^{n} u_{pqi} \le 1, \quad p=1,\dots,m , q=1,\dots,m$.
    \item $\sum_{p=1}^{m} \sum_{q=1}^{m} u_{pqi} = 1 , \quad i=1,\dots,n$.
    \item 
    $  L_p - (1 - u_{pqi}) \le 
x_i \le U_p + (1 - u_{pqi}), \ L_q - (1 - u_{pqi}) \le 
x_i \le U_q + (1 - u_{pqi}) , \quad  p,q=1,\dots,m, i=1,\dots,n  $.
\end{enumerate}
\end{cor}

In our computational experiments, we instantiate Corollary~\ref{cor:one-per-subsquare-encoding} with the specific choice of 
\(
m= \left\lfloor \frac{1}{\underline{H}_n -10^{-6} }  \right\rfloor 
\).

\begin{rem}
Proposition~\ref{prop:one-per-square-general} suggests that one could, in
principle, strengthen Corollary~\ref{cor:one-per-subsquare-encoding} by adding
further families of small squares in $[0,1]^2$, together with corresponding
binary assignment variables and capacity constraints, so as to approximate the
full geometric condition for all admissible squares. However, in our
computational experiments (see Remark~\ref{Remark3}) we observed that each such
extension substantially increases the number of binary variables and the overall
model size, which in turn led to longer solution times and no commensurate
practical benefit. For this reason, we restrict this enhancement to the single
regular $m\times m$ grid of axis-aligned subsquares used in
Corollary~\ref{cor:one-per-subsquare-encoding}, which already leverages
Proposition~\ref{prop:one-per-square-general} in a numerically effective way.
\end{rem}

\subsubsection{Applying the Heilbronn Triangle Problem in Smaller Rectangles}
In Section \ref{rec2point}, we established bounds by partitioning the square into strips small enough to contain at most two points. We now generalize this approach to larger rectangles (strips of greater height) that can contain three or more points. The core idea is to determine the \emph{capacity} of a horizontal strip of a given height. If we cannot place $m$ points into a strip without forcing the minimum triangle area to fall below our lower bound of $\underline{H}_n$, then that strip is   constrained to contain at most $m-1$ points. 


Since we have  $y_1 \le \dots \le y_n$ due to Proposition~\ref{prop:order-conventions}(i), determining the capacity of a bottom strip $[0, 1] \times [0, 1/\kappa]$ imposes explicit  bounds on some of the $y_i$ variables (here, $\kappa$ is the number of strips that equipartition vertical edge of the unit square). Specifically, if a strip of height $1/\kappa$ can contain at most $m$ points, then the $(m+1)$-th point must lie strictly above $1/\kappa$ (i.e., $y_{m+1} \ge 1/\kappa$). By systematically applying this logic to partitions of the unit square, we derive rigorous variable bounds that significantly prune the search space \emph{a priori}. Figure~\ref{fig:smaller_rectangles} illustrates the concept of solving the Heilbronn problem within such a restricted domain.

\begin{figure}[h!]
\centering
\begin{tikzpicture}[scale=4, thick]

  \def\x{3}                          
  \pgfmathsetmacro{\hrect}{1/\x}     

  \draw (0,0) rectangle (1,1);

  \draw (0,0) rectangle (1,\hrect);

  \coordinate (Q1) at (0.3339240927352503, 0.0);
  \coordinate (Q2) at (1.0, 0.0);
  \coordinate (Q3) at (0.0, 0.5778621196377186*\hrect);
  \coordinate (Q4) at (1.0, 0.6660759072647495*\hrect);
  \coordinate (Q5) at (0.4221376453891849, \hrect);

  \begin{scope}
    \clip (0,0) rectangle (1,1); 

    \draw[black!60] (Q1)--(Q2)--(Q3)--cycle;
    \draw[black!60] (Q1)--(Q2)--(Q4)--cycle;
    \draw[black!60] (Q1)--(Q2)--(Q5)--cycle;
    \draw[black!60] (Q1)--(Q3)--(Q4)--cycle;
    \draw[black!60] (Q1)--(Q3)--(Q5)--cycle; 
    \draw[black!60] (Q1)--(Q4)--(Q5)--cycle;
    \draw[black!60] (Q2)--(Q3)--(Q4)--cycle;
    \draw[black!60] (Q2)--(Q3)--(Q5)--cycle;
    \draw[black!60] (Q2)--(Q4)--(Q5)--cycle;
    \draw[black!60] (Q3)--(Q4)--(Q5)--cycle;

    \fill[red!25,opacity=.40] (Q1)--(Q3)--(Q5)--cycle;
    \draw[red]                (Q1)--(Q3)--(Q5)--cycle;
  \end{scope}

  \foreach \pt in {Q1,Q2,Q3,Q4,Q5}{\fill[black] (\pt) circle(0.02);}
  \foreach \pt in {Q1,Q3,Q5}{\fill[red] (\pt) circle(0.02);}

  \draw[<->] (1.05,0) -- (1.05,\hrect);
  \node[right, font=\scriptsize] at (1.05,0.5*\hrect) {$\dfrac{1}{\kappa}$};

  \draw[<->] (0,-0.07) -- (1,-0.07);
  \node[below, font=\scriptsize] at (0.5,-0.07) {$1$};

\end{tikzpicture}
\caption{Applying the Heilbronn triangle problem in smaller rectangles to derive cardinality constraints. The height of the restricted strip is $1/\kappa$, where $\kappa$ is the number of strips.} 
\label{fig:smaller_rectangles}
\end{figure}
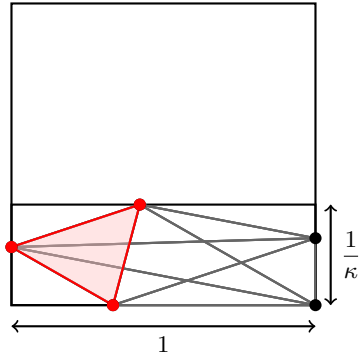

\begin{prop}\label{prop:strip-bounds}
Fix \(\kappa\ge 2\) and partition the unit square into \(\kappa\) equal horizontal strips of height \(1/\kappa\).
Let \(m_\kappa\) be the maximum number of points that can be placed in a strip with height $1/\kappa$ and width 1 such that the minimum area of any triangle formed by these points is at least \(\underline{H}_n\).
Then, there exists an optimal solution to the Heilbronn triangle problem that satisfies the following constraints: 
\[
y_{\,\ell\,m_\kappa+1}\ \ge\ \frac{\ell}{\kappa}
\qquad\text{and}\qquad
y_{\,n-\ell\,m_\kappa}\ \le\ 1-\frac{\ell}{\kappa}, \text{ for } \ell= 1,\dots,\kappa-1 .
\]
\end{prop}

\begin{proof}
Let us first consider the lower bound. The first \(\ell\) strips together form a rectangle of height \(\ell/\kappa\). If \(y_{\ell\,m_\kappa+1} < \ell/\kappa\), then the first \(\ell\,m_\kappa + 1\) points lie strictly inside these first \(\ell\) strips. By the Generalized Pigeonhole Principle, at least one of these \(\ell\) strips must contain \(\lceil (\ell\,m_\kappa + 1)/\ell \rceil = m_\kappa + 1\) points. However, by the definition of \(m_\kappa\), placing \(m_\kappa+1\) points in a strip of height \(1/\kappa\) forces the minimum triangle area to be strictly less than \(\underline{H}_n\), contradicting feasibility. Thus, we must have \(y_{\ell\,m_\kappa+1} \ge \ell/\kappa\).

For the upper bound, consider the top \(\ell\) strips, which cover the vertical range \([1 - \ell/\kappa, 1]\). Suppose, for the sake of contradiction, that \(y_{n - \ell\,m_\kappa} > 1 - \ell/\kappa\). Since the \(y\)-coordinates are sorted in a non-decreasing fashion (\(y_1 \le \dots \le y_n\)), this inequality implies that all points \(y_i\) with indices \(i \ge n - \ell\,m_\kappa\) must also lie strictly within this top region. The number of such points is \(n - (n - \ell\,m_\kappa) + 1 = \ell\,m_\kappa + 1\). We are thus attempting to place \(\ell\,m_\kappa + 1\) points into a region consisting of \(\ell\) strips. By the Generalized Pigeonhole Principle, at least one of these top strips must contain \(\lceil (\ell\,m_\kappa + 1)/\ell \rceil = m_\kappa + 1\) points. As established in the lower bound case, a single strip of height \(1/\kappa\) cannot  contain \(m_\kappa + 1\) points.
Since this leads to a contradiction, we deduce that  \(y_{n - \ell\,m_\kappa} \le 1 - \ell/\kappa\).
\end{proof}

All bounds on the ordered \(y\)-coordinates, including those implied by Proposition~\ref{onepointoneachedge} and Proposition~\ref{prop:strip-bounds}, are compiled in \autoref{tab:strip-bounds}.

\begin{rem}

The specific bounds listed in \Cref{tab:strip-bounds} were experimentally generated using an iterative procedure. We determined the capacity \(m_\kappa\) for strips of height \(1/\kappa\) ($\kappa=2,3,\dots$) by solving a sequence of restricted Heilbronn problems using the Approach 1 (MIQCP) formulation described in \Cref{Approach1}. To accelerate these subproblems, we enforced the symmetry breaking constraints defined in \Cref{sec:order-conventions}, as well as at least one point on each edge constraints detailed in \Cref{sec:onepointeachedge}. For each strip height, we incremented the number of points \(m\) until the solver proved infeasibility (i.e., 
the optimal value dropped below \(\underline{H}_n\)); the largest feasible \(m\) was recorded as \(m_\kappa\).


By applying this capacity logic symmetrically—considering rectangles extending upwards from $y=0$ and downwards from $y=1$—we established both lower and upper bounds. The range of indices \(i\) for which we computed these bounds depends on the instance size. 
For \(n=6, 7, 8\), we targeted \(i=2, \dots, n-1\). 
For \(n=9, 10\), leveraging the boundary occupancy properties defined in Proposition~\ref{prop:nine-eight-edge} (proven for \(n=9\)) and Conjecture~\ref{conj:atleast2ptsonedges} (conjectured for \(n=10\)), we restricted the search to the intermediate points \(i=3, \dots, n-2\). 
For points excluded from this search, the bounds in \Cref{tab:strip-bounds} reflect bounds derived from Corollary~\ref{cor:implementingx} and Proposition~\ref{prop:nine-eight-edge}.

This pre-computation step is highly efficient. The total wall-clock times required to generate the complete set of bounds were approximately 0.346 seconds for \(n=6\), 3.247 seconds for \(n=7\), 29.337 seconds for \(n=8\), and 35.382 seconds for \(n=9\). For \(n=10\), the time increased to 1\,293.015 seconds, 
however, this overhead remains negligible compared to the substantial reduction in the time spent by the global solver due to these variable bounds.
\end{rem}




\begin{table}[h!]
\centering
\renewcommand{\arraystretch}{1.5}
\begin{tabular}{|c|cc|cc|cc|cc|cc|}
\hline
 & \multicolumn{2}{|c|}{$n=6$} & \multicolumn{2}{|c|}{$n=7$} & \multicolumn{2}{|c|}{$n=8$} & \multicolumn{2}{|c|}{$n=9$} & \multicolumn{2}{|c|}{$n=10$} \\ 
\hline
\textbf{$i$} & $\underline{y}_i$ & $\overline{y}_i$ & $\underline{y}_i$ & $\overline{y}_i$ & $\underline{y}_i$ & $\overline{y}_i$ & $\underline{y}_i$ &$\overline{y}_i$ & $\underline{y}_i$ & $\overline{y}_i$ \\ 
\hline
1  & $0$ & $0$ & $0$ & $0$ & $0$ & $0$ & $0$ & $0$ & $0$ & $0$ \\[0.5ex]
\hline
2  & $0$ & $\tfrac{1}{2}$ 
   & $0$ & $\tfrac{1}{2}$
   & $0$ & $\tfrac{1}{2}$
   & $0$ & $\tfrac{1}{100}$
   & $0$ & $\tfrac{1}{2}$ \\
\hline
3  & $\tfrac{1}{5}$ & $\tfrac{4}{5}$
   & $\tfrac{1}{6}$ & $\tfrac{2}{3}$
   & $\tfrac{1}{7}$ & $\tfrac{1}{2}$
   & $\tfrac{1}{10}$ & $\tfrac{1}{2}$
   & $\tfrac{1}{11}$ & $\tfrac{1}{2}$ \\
\hline
4  & $\tfrac{1}{5}$ & $\tfrac{4}{5}$
   & $\tfrac{1}{6}$ & $\tfrac{5}{6}$
   & $\tfrac{1}{7}$ & $\tfrac{2}{3}$
   & $\tfrac{1}{10}$ & $\tfrac{2}{3}$
   & $\tfrac{1}{11}$ & $\tfrac{1}{2}$ \\
\hline
5  & $\tfrac{1}{2}$ & $1$
   & $\tfrac{1}{3}$ & $\tfrac{5}{6}$
   & $\tfrac{1}{3}$ & $\tfrac{6}{7}$
   & $\tfrac{1}{4}$ & $\tfrac{3}{4}$
   & $\tfrac{1}{5}$ & $\tfrac{2}{3}$ \\
\hline
6  & $1$ & $1$
   & $\tfrac{1}{2}$ & $1$
   & $\tfrac{1}{2}$ & $\tfrac{6}{7}$
   & $\tfrac{1}{3}$ & $\tfrac{9}{10}$
   & $\tfrac{1}{3}$ & $\tfrac{4}{5}$ \\
\hline
7  & -- & --
   & $1$ & $1$
   & $\tfrac{1}{2}$ & $1$
   & $\tfrac{1}{2}$ & $\tfrac{9}{10}$
   & $\tfrac{1}{2}$ & $\tfrac{10}{11}$ \\
\hline
8  & -- & --
   & -- & --
   & $1$ & $1$
   & $\tfrac{99}{100}$ & $1$
   & $\tfrac{1}{2}$ & $\tfrac{10}{11}$ \\
\hline
9  & -- & --
   & -- & --
   & -- & --
   & $1$ & $1$
   & {$\tfrac{1}{2}$} & $1$ \\
\hline
10    & -- & --
   & -- & --
   & -- & --
   & -- & --
   & $1$ & $1$ \\
\hline
\end{tabular}
\caption{Compilation of variable bounds for $y_i$  \(y_1 \le \cdots \le y_n\) for \(n=6,\dots,10\).
Entries are lower bounds $\underline{y}_i$ and upper bounds  $\overline{y}_i$  for each index \(i\).
The symbol -- indicates that the index \(i\) does not exist for that~\(n\).}
\label{tab:strip-bounds}
\end{table}

\section{Computational Experiments} \label{sec:computational_experiments}

\subsection{Computational Setup}
All experiments are run on a Dell workstation fitted with
two 48-core Intel Xeon Gold 6248R CPUs (3.0 GHz base clock, 96 cores in total) and 256 GB of RAM. The system operates on a 64-bit Microsoft Windows~11  system. Model development and execution are mostly carried out in Visual Studio Code, and optimization is performed with Gurobi Optimizer 11.0.1 and BARON 22.4.20.



\subsection{Preliminary Experiments Setup}\label{ExperimentSetup}
We evaluate the proposed methods in Section~\ref{Approaches} in three stages: (i) we first identify the best baseline among the three formulations in Section~\ref{ComparingApproaches}, (ii) we then select the faster solver between Gurobi and BARON in Section~\ref{sec:ComparingSolvers}, and (iii) we finally quantify the contribution of the enhancements developed in Section~\ref{sec:enhancements}, as reported in Section~\ref{sec:EnhancementAblation}, in terms of computational effort. Unless otherwise specified, all experiments are run under a uniform wall–clock time limit of one day per instance.

\subsubsection{Formulation Comparison}\label{ComparingApproaches}
We compare the three formulations on instances with \(n\in\{3,4,5,6\}\) using Gurobi under default settings, except for a one-day (i.e., a 86\;400-second) time limit. For Approach 3, we use its MILP relaxation with discretization parameter \(p=10\).
Proven lower bounds (LB) and upper bounds (UB) as well as and runtimes (in seconds) are summarized in Table~\ref{tab:approach_combined}.

\makegapedcells
\begin{table}[h!]
\centering
\begin{tabular}{|c|c|c|c|}
\hline
\textbf{$n$} &
\makecell{\textbf{Approach 1 (MIQCP)}\\[1pt]LB--UB \\(Time (s))} &
\makecell{\textbf{Approach 2 (QCP)}\\[1pt]LB--UB \\(Time (s))} &
\makecell{\textbf{Approach 3 (MILP Rel.)}\\[1pt]LB$^\prime$--UB \\(Time (s))} \\
\hline
3 & \makecell{0.5--0.5\\(0.019)} &
    \makecell{0.5--0.5\\(0.016)} &
    \makecell{0.5--0.5\\(0.092)} \\
\hline
4 & \makecell{0.5--0.5\\(0.005)} &
    \makecell{0.5--0.5\\(0.031)} &
    \makecell{0.5--0.5\\(0.047)} \\
\hline
5 & \makecell{0.1924508--0.1924508\\(4.538)} &
    \makecell{0.1924507--0.1924611\\(9.632)} &
    \makecell{0.1926462--0.1926462\\(27.557)} \\
\hline
6 & \makecell{0.1250012--0.1250137\\(468.535)} &
    \makecell{0.1250019--0.1250144\\(59\,920.142)} &
    \makecell{0.1252447--0.1266927\\({86\,400})} \\
\hline
\end{tabular}
\caption{Combined bounds and runtimes (seconds) for the three approaches. Note that Approach 3 is a relaxation; its objective values represent upper bounds on the true optimum. The column LB\(^\prime\) denotes the best objective found for the \emph{relaxed} problem, which may  not be a valid  lower bound for the original Heilbronn problem.}
\label{tab:approach_combined}
\end{table}
\makegapedcells

\paragraph{Accuracy vs.\ known optima.}
For \(n=3\) and \(n=4\), all three approaches recover the known global optimum \(0.5\), and in each case the lower and upper bounds coincide.

For \(n=5\), the proven global optimum is \(0.1924500\) \cite{lu1991goldberg}. Approaches 1 and 2 both reach this value to within the requested relative optimality tolerance (reporting \(0.1924508\) and \(0.1924507\), respectively). Approach 1 also certifies optimality by matching its lower and upper bounds exactly (\(0.1924508\)--\(0.1924508\)), whereas Approach~2 still reports a small residual gap.
Approach 3 returns a value of \(0.1926462\), which is strictly larger than the true optimum. This is consistent with the fact that Approach 3 is a relaxation of the original maximization problem; therefore, it provides a valid \emph{upper bound} on the optimal value but cannot guarantee a feasible placement (a valid lower bound) for the original problem.

For \(n=6\), where the global optimum is \(1/8 = 0.125\) \cite{lu1991goldberg}, Approaches 1 and 2 again produce lower bounds consistent with the known optimum (\(0.1250012\) and \(0.1250019\), respectively). In contrast, Approach 3 provides a 
bound of \(0.1252447\). As with the \(n=5\) case, this larger value reflects the relaxation gap: Approach 3 effectively overestimates the maximum area because it relaxes the non-convex constraints. Additionally, it retains a relatively wide optimality gap (\(0.1252447\)--\(0.1266927\)) within the time limit.

\paragraph{Computation time.}
Table~\ref{tab:approach_combined} shows that    the second approach is fastest when \(n=3\), followed by the first and (by a large margin) the third.
For larger instances, the first approach becomes consistently faster: at \(n=6\) it solves in about {$468.535$} seconds, while the second requires {$59\,920.142$} seconds and the third hits the {one–day} limit.
The performance gap grows sharply with \(n\). This widening gap indicates that, beyond very small instances, only the first approach remains computationally practical, while the others become prohibitively expensive and are no longer reasonable candidates for large \(n\).


\paragraph{Selection.}
Both the first and second approaches produce objective values that agree with the known global optima for \(n=3,\dots,6\), and in particular neither shows meaningful numerical deviation at the reported tolerance. Moreover, the first approach is often able to certify optimality outright by closing the lower/upper bound gap, whereas the second approach may still report a small residual gap. When runtime is taken into account, the distinction becomes sharper: although the second approach is competitive for very small instances, it becomes dramatically slower as \(n\) increases, while the first approach remains practical. The third approach is both less accurate (e.g., for \(n=5,6\)) and substantially slower. On this basis, we select Approach~1 as the baseline for all subsequent experiments.






\subsubsection{Solver Comparison}\label{sec:ComparingSolvers}

We evaluate the \emph{superior model} identified in Section~\ref{ComparingApproaches} (without enhancements) on both BARON and Gurobi to select the best performing solver  for subsequent experiments. 
Unless stated otherwise, both solvers are run under their default settings, except that a wall–clock limit of one day per instance is imposed.
The lower and upper bounds obtained,  and runtimes (in seconds) are summarized in Table~\ref{tab:solver_combined}.

\begin{table}[h!]
\centering
\begin{tabular}{|c|cc|cc|}
\hline
\multirow{2}{*}{$n$}    & \multicolumn{2}{c|}{\textbf{Gurobi}} &
      \multicolumn{2}{c|}{\textbf{BARON}} \\
 &
  {LB--UB} & {Time(s)} &
  {LB--UB} & {Time(s)} \\
\hline
3 & 0.5--0.5             &   0.021 & 0.5--0.5             &   0.064 \\
4 & 0.5--0.5             &   0.013 & 0.5--0.5             &   0.062 \\
5 & 0.1924508--0.1924508 &   4.540 & 0.1924553--0.1924562 & 285.302 \\
6 & 0.1250012--0.1250137 & 468.535 & 0.1250081--0.1992033 &  86\,400 \\
\hline
\end{tabular}
\caption{Combined solver bounds and runtimes (seconds). The value
86\,400\,s indicates that the one-day time limit was reached.}
\label{tab:solver_combined}
\end{table}

\paragraph{Accuracy vs.\ known optima.}
For \(n\in\{3,4\}\), both solvers recover the known global optimum \(0.5\), and in each case the lower and upper bounds coincide. For \(n=5\), the proven optimum is \(H_{5}^*=\sqrt{3}/9=0.192450\) \cite{lu1991goldberg}. Gurobi attains this value  within the requested relative optimality tolerance and certifies it by reporting matching bounds \(0.1924508\)--\(0.1924508\). BARON returns a slightly higher value, with bounds \(0.1924553\)--\(0.1924562\), which is above the proven optimum and does not close the optimality gap. For \(n=6\), where the exact value is \(H_{6}^*=1/8=0.125\) \cite{lu1991goldberg,dress1995heilbronn}, Gurobi reports bounds \(0.1250012\)--\(0.1250137\), which are consistent with the known optimum at the stated tolerance whereas BARON produces a much wider interval of \(0.1250081\)--\(0.1992033\). 

\paragraph{Computation time.}
For \(n\in\{3,4\}\), the runtime difference is negligible. 
Starting at \(n=5\), Gurobi is substantially faster (about \(4.54\) s vs.\ \(285.30\) s), and at \(n=6\) the gap widens dramatically: Gurobi solves in \(468.535\) s, while BARON reaches the one–day cap (reported as \(86{\,}400\) s) without certification.

\paragraph{Selection.}
Both solvers reproduce the known optimal value for very small instances, but starting at \(n=5\) their behavior diverges. Gurobi continues to match the known optimum to within the requested tolerance and is able to certify that value by closing the bound gap, while BARON begins to report values that exceed the proven optimum and retains a non-negligible gap. By \(n=6\), Gurobi remains consistent with the exact value \(1/8\), whereas BARON no longer delivers a tight bound. Taken together with the runtime trends in Table~\ref{tab:solver_combined}, these results indicate that {Gurobi} is the more dependable solver for our purposes, and we therefore use it for all subsequent computations.





\subsubsection{Enhancement Comparison}\label{sec:EnhancementAblation}

Building on the best approach identified in Section~\ref{ComparingApproaches} and the solver choice in Section~\ref{sec:ComparingSolvers}, we quantify the contribution of three enhancement groups via a full-factorial design with \(2^{3}\) experiments in Table~\ref{tab:groupcomparison_all}.

\begin{table}[h!]
\centering
\begin{tabular}{|c|c|c|c|}
\hline
\textbf{$n$} & \textbf{Groups} & \textbf{Lower Bound--Upper Bound} & \textbf{Time(s)} \\
\hline
\multirow{8}{*}{\textbf{7}}
& --      & 0.0838601--0.1599112 & 86\,400 \\
& 1       & 0.0838597--0.0838597 & 16.570 \\
& 2       & 0.0838599--0.0838599 & 218.782 \\
& 3       & 0.0838598--0.1229744 & 86\,400 \\
& 1, 2     & 0.0838601--0.0838601 & 7.537 \\
& 1, 3     & 0.0838599--0.0838599 & 29.638 \\
& 2, 3     & 0.0838600--0.0838600 & 104.361 \\
& 1, 2, 3   & 0.0838595--0.0838595 & 4.542 \\
\hline
\multirow{8}{*}{\textbf{8}}
& --      & 0.0723771--0.1684433 & 86\,400 \\
& 1       & 0.0723774--0.0723774 & 314.720 \\
& 2       & 0.0723775--0.0921830 & 86\,400 \\
& 3       & 0.0723772--0.1475525 & 86\,400 \\
& 1, 2     & 0.0723773--0.0723773 & 385.604 \\
& 1, 3     & 0.0723770--0.0723770 & 82.888 \\
& 2, 3     & 0.0723775--0.0723775 & 3\,639.555 \\
& 1, 2, 3   & 0.0723767--0.0723767 & 53.334 \\
\hline
\multirow{8}{*}{\textbf{9\footnotemark}}
& --      & --                    & 86\,400 \\
& 1       & 0.0548767--0.0640982 & 86\,400 \\
& 2       & 0.0548770--0.0893424 & 86\,400 \\
& 3       & 0.0483269--0.1803900 & 86\,400 \\
& 1, 2     & 0.0548769--0.0548769 & 3\,334.102 \\
& 1, 3     & 0.0548765--0.0689799 & 86\,400 \\
& 2, 3     & 0.0548771--0.0730435 & 86\,400 \\
& 1, 2, 3   & 0.0548767--0.0548767 & 633.331 \\
\hline
\end{tabular}
\caption{Combined LB/UB and runtimes (seconds) for enhancement groups across \(n\in\{7,8,9\}\), computed with Gurobi. An entry of 86{\,}400\,s indicates the time cap was reached.}
\label{tab:groupcomparison_all}
\end{table}
\footnotetext{For the $n=9$ instance, a one-day preliminary computation is required to establish Proposition~\ref{prop:nine-eight-edge}; thereafter, combinations containing enhancement group~2 could be used.\label{footnote}}

For \(n\in\{7,8,9\}\), Table~\ref{tab:groupcomparison_all} reports eight configurations corresponding to all on/off patterns of the three groups \(\mathcal G_{1},\mathcal G_{2},\mathcal G_{3}\): (i) none, (ii–iv) singletons, (v–vii) all pairwise combinations, and (viii) all three enabled. \emph{The groups themselves are defined uniformly for all \(n\)}; however, each group contains several constraints whose applicability depends on the number of points $n$. Thus, when a group is “on” at a given \(n\), only those constraints within the group whose eligibility conditions are satisfied are activated, while the inapplicable members of the group remain inactive.

Each configuration is evaluated by (a) wall-clock time (seconds) and (b) solution quality summarized by \(\mathrm{LB}/\mathrm{UB}\) (and the implied optimality gap). All runs use the default Gurobi settings
and a uniform wall-clock limit per instance (entries equal to the cap in Table~\ref{tab:groupcomparison_all} indicate runs that timed out).

This design allows us to estimate the \emph{main effects} of each group (marginal improvement when toggled on) and the \emph{interaction effects} among groups across \(n=7,8,9\).

\paragraph{Accuracy vs.\ known optima.}
For \(n=7\) and \(n=8\), the optimum values are approximately \(H_{7}^*=0.083859\) \cite{zeng2008heilbronn} and \(H_{8}^*=0.0723764\) \cite{dehbi2022heilbronn}. In Table~\ref{tab:groupcomparison_all}, some enhancement settings reproduce these known optima and certify them by reporting matching lower and upper bounds (LB\,=\,UB), while others do not. For \(n=7\), every configuration except the baseline with no enhancements (\texttt{--}) and the standalone group \(\mathcal G_{3}\) achieves the lower bound being equal to the upper bound at \(H_{7}^*\), i.e., certifies optimality. For \(n=8\), certification occurs for settings that include \(\mathcal G_{1}\) (namely \(\mathcal G_{1}\), \(\mathcal G_{1,2}\), \(\mathcal G_{1,3}\), and \(\mathcal G_{1,2,3}\)) and also for \(\mathcal G_{2,3}\); in contrast, \texttt{--}, \(\mathcal G_{2}\), and \(\mathcal G_{3}\) do not certify. 

For \(n=9\), the exact optimum is not known; the best published interval; is
\(0.054875999 < H_{9}^* < 0.054878314\) \cite{chen2017searching}. Among our configurations, only \(\mathcal G_{1,2}\) and \(\mathcal G_{1,2,3}\) return the lower bound   equals to the upper bound  at values consistent with that interval, thereby producing certificates. All other settings (\texttt{--}, \(\mathcal G_{1}\), \(\mathcal G_{2}\), \(\mathcal G_{3}\), \(\mathcal G_{1,3}\), \(\mathcal G_{2,3}\)) fail to certify within the allowed {time} (see Table~\ref{tab:groupcomparison_all}).

\paragraph{Computation time.}
The baseline \texttt{--} and the singleton \(\mathcal G_{3}\) consistently time out (86{\,}400\,s).
Among certifying configurations, there is a clear positive interaction: enabling all three groups \(\mathcal G_{1,2,3}\) is fastest for  each \(n\) (4.542\,s for \(n=7\); 53.334\,s for \(n=8\); 633.331\,s for \(n=9\)), improving substantially over strong pairwise settings (e.g., \(\mathcal G_{1,2}\): 7.537\,s, 385.604\,s, and {$3{\,}334.102\,$}s, respectively).
Some combinations can certify but remain slow (e.g., \(\mathcal G_{2,3}\) at \(n=8\): {$3{\,}639.555\,$}s, indicating that \(\mathcal G_{3}\) is beneficial only when paired with \(\mathcal G_{1}\) and/or \(\mathcal G_{2}\), and otherwise may hinder pruning.

\paragraph{Selection.}
Across \(n=7\) and \(n=8\), the enhancement settings that include \(\mathcal G_{1}\) are consistently able to match and certify the known global optima, while configurations that exclude \(\mathcal G_{1}\) often fail to do so. For \(n=9\), where only a bracket is available, the only settings that both respect this interval and close the gap are \(\mathcal G_{1,2}\) and \(\mathcal G_{1,2,3}\). When runtime is also taken into account (Table~\ref{tab:groupcomparison_all}), \(\mathcal G_{1,2,3}\) not only certifies whenever certification is possible, but does so more efficiently than the other certifying combinations. We therefore adopt \(\mathcal G_{1,2,3}\) as the default configuration for subsequent experiments.

\subsection{Final Results}\label{FinalResults}
Guided by the results above, we adopt {Approach~1} as the base formulation, use {Gurobi} as the solver, and start from the full enhancement bundle \(\mathcal{G}_{1,2,3}\) (the best-performing group combination). Instance-specific enhancements are activated {only} when valid for the given \(n\): the objective-value bound in Section~\ref{sec:boundsHn} is inapplicable at \(n=3\); the one-point-on-each-edge result (Proposition~\ref{onepointoneachedge}) also does not apply at \(n=3\); the two-points-per-edge motif (Proposition~\ref{2pointsedge}) is used only for \(n\in\{7,8\}\); and the near-boundary pattern for \(n=9\) (Section~\ref{sec:nine-eight-edge}) is used only at \(n=9\).
The performance of the resulting setting for \(3\le n\le 10\) is summarized in Table~\ref{tab:optimization_results}, which reports, for each \(n\), the certified lower/upper bounds, the wall-clock time, and the \emph{computed area}, which is computed directly from the reported coordinates of an optimal placement.

Since the optimal values for \(n = 5,6,7,8\) are already established in the literature, our figures for these rows essentially coincide with the known values; any tiny discrepancies stem from numerical tolerances in the solver.

The only previous record for the \(n = 9\) case is due to Chen et al., who reported $0.054875999 \;<\; H_{9}^* \;<\; 0.054878314$ in~\cite{chen2017searching}. Their result was the product of an intensive grid search method run on a specialized, large-scale computing infrastructure, including an IBM NUMA cluster with 4TB of RAM and three GPGPU machines equipped with 12 NVIDIA Tesla C2050 cards in total. This entire search required {$2{\,}695{\,}833.174$} seconds (approximately 31 days) of GPGPU computation time~\cite{chen2017searching}. In stark contrast, our refined mathematical programming model achieves a certified global optimum (as shown in \Cref{tab:optimization_results}) in just \(633.331\) seconds (approximately 10 minutes) when running on a single
CPU-based workstation. We would like to remind the reader that even if we include the pre-processing time needed to numerically prove Proposition~\ref{prop:nine-eight-edge} (also see, Footnote~\ref{footnote}), the resulting computational effort is still 
significantly smaller compared to reference~\cite{chen2017searching}.

\begin{table}[h!]
\centering
\begin{tabular}{|c|c|c|c|c|c|}
\hline
\textbf{$n$} & \textbf{Lower Bound} & \textbf{Upper Bound} & \textbf{Time (s)} & \textbf{Computed Area\footnotemark} & \textbf{Proven Answer}\\
\hline
3 & 0.5000000 & 0.5000000 & 0.016 & 0.5000000 & 0.5 \\
4 & 0.5000000 & 0.5000000 & 0.026 & 0.5000000 & 0.5 \\
5 & 0.1924506 & 0.1924506 & 0.186 & 0.1924499 & 0.1924500 \\
6 & 0.1250010 & 0.1250010 & 0.787 & 0.1249999 & 0.125 \\
7 & 0.0838594 & 0.0838594 & 5.280 & 0.0838584 & 0.0838591 \\
8 & 0.0723767 & 0.0723767 & 53.334 & 0.0723758 & 0.0723764 \\
9 & 0.0548767 & 0.0548767 & 633.331 & 0.0548756 & -\\
10 & 0.0465383 & - & 86\;400 & 0.0465369 & -\\
\hline
\end{tabular}
\caption{Optimization results for varying $n$ values.}
\label{tab:optimization_results}
\end{table}
\footnotetext{The solver-reported objective value is subject to the solver’s numerical tolerances and model approximations. To obtain a more accurate value for the specific configuration shown in Table~\ref{tab:n6_to_n10_placements}, we recomputed the minimum triangle area directly from the listed coordinates and report it here as the \emph{computed area}.}

The last row of \Cref{tab:optimization_results} shows our best result for $n=10$ run under   a \textit{restricted model} that included two strong, but unproven structural conjectures: 
(i) We assume that Conjecture~\ref{conj:atleast2ptsonedges} holds true, that is,   an optimal solution must place \emph{exactly} eight points on the boundary edge. This hypothesis is motivated by the   two-point-per-edge \emph{near-boundary} pattern observed for $n=9$ in \Cref{sec:nine-eight-edge}. 
(ii) We add the symmetry-breaking constraint \(y_5 \le \tfrac{1}{2}\), justified by the horizontal partition bounds in \Cref{tab:strip-bounds} that limit how many points can lie in each half.
Under this restricted model, \(n=10\) quickly exposed our current hardware ceiling. Across multiple runs with progressively tighter cuts, the solver consistently certified a \emph{provable} lower bound of \(0.0465383\) (with a computed area of~\(0.0465369\)) in about \(800\) seconds—already matching the best empirical value of \(0.046537\) reported by  Comellas and Yebra in~\cite{comellas2002new}. Closing the remaining gap to obtain a matching upper bound (and thus a full certificate of optimality) appears to require substantially longer wall time or new structural insights. While \(n=10\) remains open, the speed at which such a strong lower bound is achieved highlights the promise of the strengthened formulation once additional computational power or stronger valid inequalities   are available.


The best configurations found for \(n=6\) through \(n=10\) are visualized in Figure~\ref{fig:n6_to_n10_placements}. The complete point coordinates are also provided in Table~\ref{tab:n6_to_n10_placements} of Appendix~\ref{Appendix}.

\begin{figure}[h!]
\centering
\begin{tikzpicture}[scale=4, thick]

\tikzset{alltris/.style={black!40, line width=0.2pt}}

\begin{scope}[shift={(0,-1.35)}]
  \node[align=center] at (0.5,1.10)
    {{$n=6$}\\[-0.2ex]{\scriptsize computed area: $0.1249999$}};

  \draw (0,0) rectangle (1,1);

  \coordinate (P1) at (0.0, 0.5002079002445663);
  \coordinate (P2) at (0.9583794963842281, 0.000009489865608419215);
  \coordinate (P3) at (1.0, 0.5002257602104074);
  \coordinate (P4) at (0.4585968225886384, 0.0);
  \coordinate (P5) at (0.0415841747655408, 0.9999908730285099);
  \coordinate (P6) at (0.5418012832835395, 1.0);
  \def\N{6}

  \foreach \i in {1,...,\N}{
    \foreach \j in {1,...,\N}{
      \foreach \k in {1,...,\N}{
        \ifnum\i<\j\relax\ifnum\j<\k\relax
          \draw[alltris] (P\i)--(P\j)--(P\k)--cycle;
        \fi\fi
      }
    }
  }

  \fill[blue!25,opacity=.35] (P2)--(P3)--(P4)--cycle;
  \draw[blue!60!black]       (P2)--(P3)--(P4)--cycle;

  \foreach \Q in {P1,P2,P3,P4,P5,P6}{\fill (\Q) circle (0.015);}
\end{scope}

\begin{scope}[shift={(1.35,-1.35)}]
  \node[align=center] at (0.5,1.10)
    {{$n=7$}\\[-0.2ex]{\scriptsize computed area: $0.0838584$}};

  \draw (0,0) rectangle (1,1);

  \coordinate (P1) at (0.0, 0.0);
  \coordinate (P2) at (0.8191740916746527, 0.0);
  \coordinate (P3) at (0.41614167326405116, 0.2872578887602705);
  \coordinate (P4) at (1.0, 0.2872582794565236);
  \coordinate (P5) at (0.507413965213032, 0.8060633265680573);
  \coordinate (P6) at (0.8648098677507948, 1.0);
  \coordinate (P7) at (0.0, 1.0);
  \def\N{7}

  \foreach \i in {1,...,\N}{
    \foreach \j in {1,...,\N}{
      \foreach \k in {1,...,\N}{
        \ifnum\i<\j\relax\ifnum\j<\k\relax
          \draw[alltris] (P\i)--(P\j)--(P\k)--cycle;
        \fi\fi
      }
    }
  }

  \fill[teal!25,opacity=.35] (P2)--(P4)--(P6)--cycle;
  \draw[teal!80!black]       (P2)--(P4)--(P6)--cycle;

  \foreach \Q in {P1,P2,P3,P4,P5,P6,P7}{\fill (\Q) circle (0.015);}
\end{scope}

\begin{scope}[shift={(2.70,-1.35)}]
  \node[align=center] at (0.5,1.10)
    {{$n=8$}\\[-0.2ex]{\scriptsize computed area: $0.0723758$}};

  \draw (0,0) rectangle (1,1);

  \coordinate (P1) at (0.0, 0.0);
  \coordinate (P2) at (0.8114202566960157, 0.0);
  \coordinate (P3) at (1.0, 0.2324081561857156);
  \coordinate (P4) at (0.37716169952048906, 0.23240815618571525);
  \coordinate (P5) at (0.0, 0.7675903042903774);
  \coordinate (P6) at (0.6228391395467414, 0.7675911648645654);
  \coordinate (P7) at (0.18858118900903162, 1.0);
  \coordinate (P8) at (1.0, 1.0);
  \def\N{8}

  \foreach \i in {1,...,\N}{
    \foreach \j in {1,...,\N}{
      \foreach \k in {1,...,\N}{
        \ifnum\i<\j\relax\ifnum\j<\k\relax
          \draw[alltris] (P\i)--(P\j)--(P\k)--cycle;
        \fi\fi
      }
    }
  }

  \fill[purple!25,opacity=.35] (P3)--(P6)--(P7)--cycle;
  \draw[purple!80!black]       (P3)--(P6)--(P7)--cycle;

  \foreach \Q in {P1,P2,P3,P4,P5,P6,P7,P8}{\fill (\Q) circle (0.015);}
\end{scope}


\begin{scope}[shift={(0.675,-2.70)}]
  \node[align=center] at (0.5,1.10)
    {{$n=9$}\\[-0.2ex]{\scriptsize computed area: $0.0548757$}};

  \draw (0,0) rectangle (1,1);

  \coordinate (P1) at (0.1734433903651553, 0.0);
  \coordinate (P2) at (0.8062260941999951, 0.0);
  \coordinate (P3) at (1.0, 0.17344395573013657);
  \coordinate (P4) at (0.0, 0.17344394237353644);
  \coordinate (P5) at (0.6531127134403488, 0.6531128229910201);
  \coordinate (P6) at (1.0, 0.7398341002188249);
  \coordinate (P7) at (0.0, 0.806226928561121);
  \coordinate (P8) at (0.17344290319510058, 1.0);
  \coordinate (P9) at (0.7398354599600717, 1.0);
  \def\N{9}

  \foreach \i in {1,...,\N}{
    \foreach \j in {1,...,\N}{
      \foreach \k in {1,...,\N}{
        \ifnum\i<\j\relax\ifnum\j<\k\relax
          \draw[alltris] (P\i)--(P\j)--(P\k)--cycle;
        \fi\fi
      }
    }
  }

  \fill[orange!25,opacity=.35] (P1)--(P5)--(P9)--cycle;
  \draw[orange!80!black]       (P1)--(P5)--(P9)--cycle;

  \foreach \Q in {P1,P2,P3,P4,P5,P6,P7,P8,P9}{\fill (\Q) circle (0.015);}
\end{scope}

\begin{scope}[shift={(2.025,-2.70)}]
  \node[align=center] at (0.5,1.10)
    {{$n=10$}\\[-0.2ex]{\scriptsize computed area: $0.0465369$}};

  \draw (0,0) rectangle (1,1);

  \coordinate (P1) at (0.15768906661548968, 0.0);
  \coordinate (P2) at (0.7479323353900781, 0.0);
  \coordinate (P3) at (0.0, 0.1576884322844123);
  \coordinate (P4) at (1.0, 0.2520799751572953);
  \coordinate (P5) at (0.6846219009192053, 0.31538516443587444);
  \coordinate (P6) at (0.3153850579539533, 0.684620952702754);
  \coordinate (P7) at (0.0, 0.7479304762954191);
  \coordinate (P8) at (1.0, 0.8423081336124026);
  \coordinate (P9) at (0.8423085297469296, 1.0);
  \coordinate (P10) at (0.2520782200363807, 1.0);
  \def\N{10}

  \foreach \i in {1,...,\N}{
    \foreach \j in {1,...,\N}{
      \foreach \k in {1,...,\N}{
        \ifnum\i<\j\relax\ifnum\j<\k\relax
          \draw[alltris] (P\i)--(P\j)--(P\k)--cycle;
        \fi\fi
      }
    }
  }

  \fill[red!25,opacity=.35] (P4)--(P8)--(P9)--cycle;
  \draw[red!80!black]       (P4)--(P8)--(P9)--cycle;

  \foreach \Q in {P1,P2,P3,P4,P5,P6,P7,P8,P9,P10}{\fill (\Q) circle (0.015);}
\end{scope}

\end{tikzpicture}
\caption{Point placements for $n=6$ to $n=10$. One of the smallest triangles is highlighted.}
\label{fig:n6_to_n10_placements}
\end{figure}
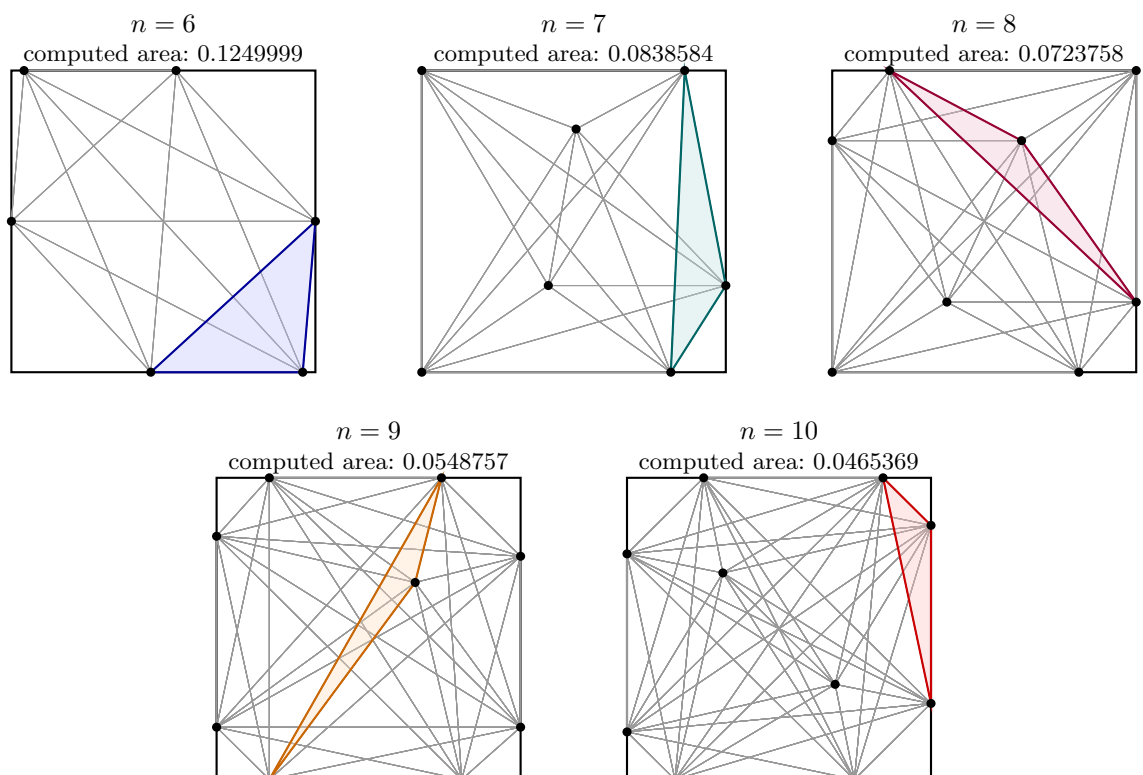

\section{Conclusions}\label{sec:conclusions}
In this paper, we explored three novel modeling approaches for the Heilbronn triangle problem and systematically strengthened each one with several groups of enhancements. Every group was first tested on small instances, both on its own and in combination with the others, and only the combinations that consistently reduced run‑time while preserving a proof of optimality were retained.

After experimenting with various  modeling approaches, we also benchmarked alternative solvers.  In particular, a BARON implementation   was tested against a Gurobi setup; the comparison showed that Gurobi consistently delivered shorter runtimes and tighter bounds.  With the optimal combination of constraints and the more efficient solver in place, the results were consolidated into Table~\ref{tab:optimization_results}.

With this configuration in place, we were able to \emph{certify} the global optimum for the nine–point case: the solver returned matching bounds (LB\,=\,UB) at \(0.0548767\) for \(n=9\), thereby closing the previously published bracket from~\cite{chen2017searching} and providing a full certificate of optimality (see Table~\ref{tab:optimization_results}). For \(n=10\), full certification remains open: under a restricted model that imposes two unproven structural assumptions—(i) the “two points per edge” near–boundary pattern and (ii) the symmetry–breaking constraint \(y_5 \le \tfrac{1}{2}\)—we obtained a strong valid lower bound of \(0.0465383\) (with a computed  area of  \(0.0465369\)) in about \(800\) seconds; however, no matching upper bound was reached within the one–day limit. Further structural insights or longer runs will likely be required to settle \(n=10\).

Looking ahead, we see a concrete path to certification for larger instances. First, the structural phenomenon behind Proposition~\ref{sec:nine-eight-edge} (“approximately two points per edge”) appears to extend to \(n=10\); a formal proof of this near-boundary pattern might materially shrink the feasible region and could close the remaining gap. Second, we have developed additional symmetry-breaking rules 
that, while not yet proved, consistently tighten relaxations in practice and are strong candidates for formal validation. Third, we plan to augment the model with \emph{orientation-count} constraints: letting \(\beta_{ijk}\in\{0,1\}\) indicate the sign of the signed area for triangle \((i,j,k)\), one can bound  the total number of “positive’’ vs.\ “negative’’ triangles, thereby restricting combinatorial degeneracies that currently inflate the search space. Together with tighter versions of existing cuts, these directions make the framework technically promising for certifying \(n=10\) and pushing toward \(n=11\).



\section*{Declarations}

\subsection*{Ethics approval and consent to participate}
Not applicable

\subsection*{Consent for publication}
Not applicable

\subsection*{Funding}
Not applicable


\subsection*{Availability of data and materials}
Data is provided within the manuscript and references therein.

\subsection*{Competing interests}
The authors declare that they have no competing interests.

\subsection*{Authors' contributions}
Amirali Modir and Amirhossein Monji worked on the methodology, implementation, visualization, analysis, writing and Burak Kocuk worked on the methodology, analysis, writing.

\bibliography{sn-bibliography}

@article{McCormick1976,
	author = {McCormick, Garth P. },
	date = {1976/12/01},
	id = {McCormick1976},
	isbn = {1436-4646},
	journal = {Mathematical Programming},
	number = {1},
	pages = {147--175},
	title = {Computability of global solutions to factorable nonconvex programs: Part I ---Convex underestimating problems},
	volume = {10},
	year = {1976}}

@article{baron,
	author = {Tawarmalani, Mohit and Sahinidis, Nikolaos V.},
	isbn = {1436-4646},
	journal = {Mathematical Programming},
	number = {2},
	pages = {225--249},
	title = {A polyhedral branch-and-cut approach to global optimization},
	volume = {103},
	year = {2005}}

@misc{Gurobi,
  author = {Gurobi Optimization, LLC},
  title = {{Gurobi Optimizer Reference Manual}},
  year = 2022,
  howpublished = {\url{www.gurobi.com/documentation/}}}

@article{chen2017searching,
  title={Searching approximate global optimal Heilbronn configurations of nine points in the unit square via GPGPU computing},
  author={Chen, Liangyu and Xu, Yaochen and Zeng, Zhenbing},
  journal={Journal of Global Optimization},
  volume={68},
  pages={147--167},
  year={2017},
  publisher={Springer}
}

@article{comellas2002new,
  title={New lower bounds for Heilbronn numbers},
  author={Comellas, Francesc and Yebra, J Luis A},
  journal={the electronic journal of combinatorics},
  pages={R6--R6},
  year={2002}
}

@article{dehbi2022heilbronn,
  title={Heilbronn’s Problem of Eight Points in the Square},
  author={Dehbi, Lydia and Zeng, Zhenbing},
  journal={Journal of Systems Science and Complexity},
  volume={35},
  number={6},
  pages={2452--2480},
  year={2022},
  publisher={Springer}
}

@inproceedings{zeng2008heilbronn,
  title={On the Heilbronn optimal configuration of seven points in the square},
  author={Zeng, Zhenbing and Chen, Liangyu},
  booktitle={International Workshop on Automated Deduction in Geometry},
  pages={196--224},
  year={2008},
  organization={Springer}
}

@article{tacspinar2024discretization,
  title={Discretization-based solution approaches for the circle packing problem},
  author={Ta{\c{s}}p{\i}nar, Rabia and Kocuk, Burak},
  journal={Engineering Optimization},
  volume={56},
  number={12},
  pages={2060--2077},
  year={2024},
  publisher={Taylor \& Francis}
}

@article{jalilian2023improved,
  title={Improved rank-one-based relaxations and bound tightening techniques for the pooling problem},
  author={Jalilian, Mosayeb and Kocuk, Burak},
  journal={Optimization and Engineering},
  pages={1--51},
  year={2025},
  publisher={Springer}
}

@book{lu1991goldberg,
  title={On Goldberg's conjecture: computing the first several Heilbronn numbers},
  author={Lu, Yang and Jingzhong, Zhang and Zhenbing, Zeng},
  year={1991},
  publisher={Universit{\"a}t Bielefeld. SFB 343. Diskrete Strukturen in der Mathematik}
}

@incollection{dress1995heilbronn,
  title={Heilbronn problem for six points in a planar convex body},
  author={Dress, Andreas W. M. and Yang, Lu and Zeng, Zhenbing},
  booktitle={Minimax and Applications},
  pages={173--190},
  year={1995},
 address   = {Boston, MA, USA} ,
  publisher={Springer}
}

@article{goldberg1972maximizing,
  title={Maximizing the smallest triangle made by N points in a square},
  author={Goldberg, Michael},
  journal={Mathematics Magazine},
  volume={45},
  number={3},
  pages={135--144},
  year={1972},
  publisher={Taylor \& Francis}
}

@article{jiang2002average,
  title={The average-case area of Heilbronn-type triangles},
  author={Jiang, Tao and Li, Ming and Vit{\'a}nyi, Paul},
  journal={Random Structures \& Algorithms},
  volume={20},
  number={2},
  pages={206--219},
  year={2002},
  publisher={Wiley Online Library}
}

@inproceedings{jiang1999expected,
  title={The expected size of Heilbronn's triangles},
  author={Jiang, Tao and Li, Ming and Vit{\'a}nyi, Paul},
  booktitle={Proceedings. Fourteenth Annual IEEE Conference on Computational Complexity (Formerly: Structure in Complexity Theory Conference)(Cat. No. 99CB36317)},
  pages={105--113},
  year={1999},
  organization={IEEE}
}

@article{bertram2000algorithm,
  title={An algorithm for Heilbronn's problem},
  author={Bertram--Kretzberg, Claudia and Hofmeister, Thomas and Lefmann, Hanno},
  journal={SIAM Journal on Computing},
  volume={30},
  number={2},
  pages={383--390},
  year={2000},
  publisher={SIAM}
}

@article{komlos1982lower,
  title={A lower bound for Heilbronn's problem},
  author={Koml{\'o}s, J{\'a}nos and Pintz, J{\'a}nos and Szemer{\'e}di, Endre},
  journal={Journal of the London Mathematical Society},
  volume={2},
  number={1},
  pages={13--24},
  year={1982},
  publisher={Oxford University Press}
}

@article{komlos1981heilbronn,
  title={On Heilbronn's triangle problem},
  author={Koml{\'o}s, J{\'a}nos and Pintz, J{\'a}nos and Szemer{\'e}di, Endre},
  journal={Journal of the London Mathematical Society},
  volume={2},
  number={3},
  pages={385--396},
  year={1981},
  publisher={Oxford University Press}
}

@article{roth1976developments,
  title={Developments in Heilbronn's triangle problem},
  author={Roth, Klaus F.},
  journal={Advances in Mathematics},
  volume={22},
  number={3},
  pages={364--385},
  year={1976},
  publisher={Elsevier}
}

@article{roth19722problem,
  title={On a problem of Heilbronn, III},
  author={Roth, Klaus F.},
  journal={Proceedings of the London Mathematical Society},
  volume={3},
  number={3},
  pages={543--549},
  year={1972},
  publisher={Oxford University Press}
}

@article{roth1972problem,
  title={On a problem of Heilbronn, II},
  author={Roth, Klaus F.},
  journal={Proceedings of the London Mathematical Society},
  volume={3},
  number={2},
  pages={193--212},
  year={1972},
  publisher={Oxford University Press}
}

@article{roth1951problem,
  title={On a problem of Heilbronn},
  author={Roth, Klaus F.},
  journal={Journal of the London Mathematical Society},
  volume={1},
  number={3},
  pages={198--204},
  year={1951},
  publisher={Wiley Online Library}
}

@article{dey2020convexifications,
  title={Convexifications of rank-one-based substructures in QCQPs and applications to the pooling problem},
  author={Dey, Santanu S. and Kocuk, Burak and Santana, Asteroide},
  journal={Journal of Global Optimization},
  volume={77},
  number={2},
  pages={227--272},
  year={2020},
  publisher={Springer}
}

@article{dey2015analysis,
  title={Analysis of MILP techniques for the pooling problem},
  author={Dey, Santanu S. and Gupte, Akshay},
  journal={Operations Research},
  volume={63},
  number={2},
  pages={412--427},
  year={2015},
  publisher={INFORMS}
}

@inbook{MaxNumber_Litvinchev2015_general,
author = {Litvinchev, I. and Infante, L. and Ozuna Espinosa, E.},
year = {2015},
month = {09},
pages = {117-135},
title = {Approximate Packing: Integer Programming Models, Valid Inequalities and Nesting},
volume = {105},
publisher = {Springer},
isbn = {978-3-319-18898-0},
doi = {10.1007/978-3-319-18899-7_9}
}

\begin{appendices}

\section{Coordinates for $n=6$ to $n=10$}\label{Appendix}

All coordinates lie in the unit square $[0,1]^2$. For each $n$, the table lists the point coordinates
$\{(x_i,y_i)\}_{i=1}^n$. 
For each configuration, we also report the value \(H_n\), computed directly from those coordinates as the minimum area among all \(\binom{n}{3}\) triangles determined by the points. In other words, \(H_n\) is the smallest triangle area achieved by the explicit placement shown in the table.

\setlength{\tabcolsep}{3pt}
\begin{table}[h]
\centering
\scriptsize
\begin{tabular}{|c|cc|cc|cc|cc|cc|}
\hline
 & \multicolumn{2}{c|}{$n=6$ (0.1249999)} & \multicolumn{2}{c|}{$n=7$ (0.0838584)} & \multicolumn{2}{c|}{$n=8$ (0.0723758)} & \multicolumn{2}{c|}{$n=9$ (0.0548757)} & \multicolumn{2}{c|}{$n=10$ (0.0465369)} \\
\hline
\textbf{$i$}
  & \textbf{$x_i$} & \textbf{$y_i$}
  & \textbf{$x_i$} & \textbf{$y_i$}
  & \textbf{$x_i$} & \textbf{$y_i$}
  & \textbf{$x_i$} & \textbf{$y_i$}
  & \textbf{$x_i$} & \textbf{$y_i$} \\
\hline
1  & 0.0000000 & 0.5002079 & 0.0000000 & 0.0000000 & 0.0000000 & 0.0000000 & 0.1734434 & 0.0000000 & 0.1576891 & 0.0000000 \\
\hline
2  & 0.9583795 & 0.0000095 & 0.8191741 & 0.0000000 & 0.8114203 & 0.0000000 & 0.8062261 & 0.0000000 & 0.7479323 & 0.0000000 \\
\hline
3  & 1.0000000 & 0.5002258 & 0.4161417 & 0.2872579 & 1.0000000 & 0.2324082 & 1.0000000 & 0.1734440 & 0.0000000 & 0.1576884 \\
\hline
4  & 0.4585968 & 0.0000000 & 1.0000000 & 0.2872583 & 0.3771617 & 0.2324082 & 0.0000000 & 0.1734439 & 1.0000000 & 0.2520800 \\
\hline
5  & 0.0415842 & 0.9999909 & 0.5074140 & 0.8060633 & 0.0000000 & 0.7675903 & 0.6531127 & 0.6531128 & 0.6846219 & 0.3153852 \\
\hline
6  & 0.5418013 & 1.0000000 & 0.8648099 & 1.0000000 & 0.6228391 & 0.7675912 & 1.0000000 & 0.7398341 & 0.3153851 & 0.6846210 \\
\hline
7  & -- & -- & 0.0000000 & 1.0000000 & 0.1885812 & 1.0000000 & 0.0000000 & 0.8062269 & 0.0000000 & 0.7479305 \\
\hline
8  & -- & -- & -- & -- & 1.0000000 & 1.0000000 & 0.1734429 & 1.0000000 & 1.0000000 & 0.8423081 \\
\hline
9  & -- & -- & -- & -- & -- & -- & 0.7398355 & 1.0000000 & 0.8423085 & 1.0000000 \\
\hline
10 & -- & -- & -- & -- & -- & -- & -- & -- & 0.2520782 & 1.0000000 \\
\hline
\end{tabular}
\caption{The exact coordinates of the point placements and the corresponding minimum  area from Figure~\ref{fig:n6_to_n10_placements} for different $n$ values.}
\label{tab:n6_to_n10_placements}
\end{table}

\end{appendices}

\end{document}